\documentclass[12pt]{iopart}

\usepackage{graphics}
\usepackage{graphicx}
\usepackage{color, hyperref}
\usepackage{cite}
\usepackage{amsthm}
\usepackage{amsmath, amssymb, latexsym, bm, mathtools, braket, multirow, url}

\theoremstyle{plain}
\newtheorem{theorem}{Theorem}[section]
\newtheorem{definition}[theorem]{Definition}
\newtheorem{proposition}[theorem]{Proposition}
\newtheorem{corollary}[theorem]{Corollary}
\newtheorem{lemma}[theorem]{Lemma}

\theoremstyle{definition}
\newtheorem{remark}[theorem]{Remark}
\newtheorem{example}[theorem]{Example}

\newcounter{count}
\newcommand{\num}{\stepcounter{count}\the\value{count}}

\newcommand{\ketbra}[2]{\ket{#1}\!\bra{#2}}
\newcommand{\abs}[1]{\left|#1\right|}
\newcommand{\conv}{\operatorname{conv}}
\newcommand{\vspan}{\operatorname{span}}
\newcommand{\Ker}{\operatorname{Ker}}
\newcommand{\Ran}{\operatorname{Ran}}

\newcommand{\Sep}{\operatorname{Sep}}

\newcommand{\am}{\operatorname{am}}
\newcommand{\gm}{\operatorname{gm}}
\newcommand{\tint}{\operatorname{int}}
\newcommand{\cl}{\operatorname{cl}}
\DeclareMathOperator{\id}{id}

\newcommand{\cB}{\mathcal{B}}

\newcommand{\cD}{\mathcal{D}}

\newcommand{\cH}{\mathcal{H}}

\newcommand{\cK}{\mathcal{K}}
\newcommand{\cL}{\mathcal{L}}

\newcommand{\cO}{\mathcal{O}}

\newcommand{\cS}{\mathcal{S}}
\newcommand{\cT}{\mathcal{T}}

\newcommand{\cV}{\mathcal{V}}
\newcommand{\cW}{\mathcal{W}}
\newcommand{\cX}{\mathcal{X}}

\begin{document}

\title[Asymptotic Properties for Markovian Dynamics in Quantum Theory and GPTs]{Asymptotic Properties for Markovian Dynamics in Quantum Theory and General Probabilistic Theories}

\author{Yuuya Yoshida$^1$ and Masahito Hayashi$^{1,2,3}$}

\address{$^1$ Graduate School of Mathematics, Nagoya University, Nagoya, Japan}
\address{$^2$ Shenzhen Institute for Quantum Science and Engineering, Southern University
of Science and Technology, Nanshan District, Shenzhen 518055,
People's Republic of China}
\address{$^3$ Centre for Quantum Technologies, National University of Singapore, 3 Science 
Drive 2, 117542, Singapore}
\eads{\mailto{m17043e@math.nagoya-u.ac.jp}, \mailto{masahito@math.nagoya-u.ac.jp}}
\vspace{10pt}
\begin{indented}
\item[]June 2019
\end{indented}

\begin{abstract}
We address \textit{asymptotic decoupling} in the context of Markovian quantum dynamics. 
Asymptotic decoupling is an asymptotic property on a bipartite quantum system, 
and asserts that the correlation between two quantum systems is broken after a sufficiently long time passes. 
The first goal of this paper is to show that asymptotic decoupling is equivalent to \textit{local mixing} 
which asserts the convergence to a unique stationary state on at least one quantum system. 
In the study of Markovian dynamics, 
mixing and ergodicity are fundamental properties 
which assert the convergence and the convergence of the long-time average, respectively. 
The second goal of this paper is to show that 
mixing for dynamics is equivalent to ergodicity for the two-fold tensor product of dynamics. 
This equivalence gives us a criterion of mixing that is a system of linear equations. 
All results in this paper are proved in the framework of general probabilistic theories (GPTs), 
but we also summarize them in quantum theory.
\end{abstract}

\vspace{2pc}


\section{Introduction}\label{introduction}
\textit{Decoupling} asserts that 
a quantum channel breaks the correlation between two quantum systems, 
and attracts attention 
in open quantum systems \cite{VKL} and quantum information \cite{MBDRC,Dupuis,DBWR}. 
For example, Dupuis et al.\ \cite{DBWR} proved a decoupling theorem 
and clarified a relation between the accuracy of decoupling and conditional entropies. 
Since the correlation between two quantum systems should be small 
for instance in evaluating information leakage \cite{Nielsen-Chuang,Hayashi}, 
it is an interesting topic to investigate decoupling in the context of quantum dynamics. 
In this paper, 
we introduce \textit{asymptotic decoupling} in the context of Markovian quantum dynamics, 
and clarify a necessary and sufficient condition of asymptotic decoupling. 
By our necessary and sufficient condition, 
asymptotic decoupling is closely related to another fundamental property 
in the study of Markovian dynamics.

Markovian dynamics has been often discussed 
in the context of statistical mechanics \cite{Streater,BP,BF}. 
In such studies, it is important how an initial state changes after a sufficiently long time passes. 
In particular, the convergence (called \textit{mixing}) 
and the convergence of the long-time average (called \textit{ergodicity}) interest many researchers \cite{BG,Burgarth}. 
More precisely, mixing asserts that, for any initial state, 
the state at discrete time $n$ (continuous time $t$) converges to a unique stationary state 
as $n\to\infty$ ($t\to\infty$); 
ergodicity asserts that, for any initial state, 
the long-time average converges to a unique stationary state. 
Hence mixing and ergodicity have been discussed in the context of 
relaxation to thermal equilibrium \cite{Streater,BP,BF}. 
However, beyond the study of relaxation processes, 
these properties have found important applications in several fields of quantum information theory. 
In particular, in quantum control \cite{WBKH}, quantum estimation \cite{Guta}, 
quantum communication \cite{GB06}, 
and the study of efficient tensorial representations 
of critical many-body quantum systems \cite{FNW}. 
Ergodicity can be checked by solving a system of linear equations \cite{BG,Burgarth}, 
but a well-known criterion of mixing requires us 
to solve an eigenequation \cite{BG}. 
Thus mixing is clearly different from ergodicity 
with respect to their definitions and numerical checks. 
In this paper, we show that mixing for dynamics is equivalent to 
ergodicity for the two-fold tensor product of dynamics. 
It enables us to check mixing by solving a system of linear equations. 
Furthermore, asymptotic decoupling is equivalent to \textit{local mixing} 
which asserts the convergence to a unique stationary state on at least one quantum system.

As a simple application of the above result, 
we give a relation between \textit{irreducibility} and \textit{primitivity} 
which are important properties in Perron-Frobenius theory. 
Irreducibility and primitivity guarantee the existence of Perron-Frobenius eigenvalues, 
which play an important roll 
in analyzing hidden Markovian processes \cite{W-H,HY}. 
For example, Perron-Frobenius eigenvalues characterize 
the asymptotic performance of the average of observed values 
in a classical hidden Markovian process: 
the central limit theorem, large deviation, and moderate deviation \cite{W-H}. 
The same method can be applied to a hidden Markovian process 
with a quantum hidden system \cite{HY}. 
For the above importance, we address irreducibility and primitivity. 
Since irreducibility and primitivity are close to ergodicity and mixing respectively, 
irreducibility for dynamics is equivalent to primitivity for the two-fold tensor product of dynamics. 
Using this equivalence, we obtain many conditions equivalent to primitivity, 
since there are many conditions equivalent to irreducibility.

In general, quantum dynamics is composed of quantum channels 
given as trace-preserving and completely positive linear maps (TP-CP maps). 
However, in order to handle stochastic operation including 
stochastic local operation and classical communication (SLOCC), 
we need to address quantum channels without trace-preserving. 
Fortunately, most results in this paper do not require trace-preserving. 
Also, complete positivity is not necessarily required and positivity is required. 
Therefore, we show our results for positive maps mainly.

So far, we have stated the asymptotic properties in the context of Markovian classical/quantum dynamics, 
but our results hold in the framework of \textit{general probabilistic theories (GPTs)}. 
GPTs are a general framework including quantum theory and classical probability theory \cite{Short-Wehner,Janotta-Hinrichsen,Lami}. 
Although quantum theory is widely accepted in current physics, 
more general requirements based on states and measurements only imply GPTs 
and do not imply quantum theory uniquely. 
Hence some researchers in physics study GPTs to explore conditions characterizing quantum theory. 
Since our results hold in GPTs, 
the asymptotic properties for Markovian dynamics do not require the framework of quantum theory 
and are common properties in GPTs. 
GPTs enable us to be easily address Markovian quantum dynamics composed of (not necessarily CP) positive maps. 
Therefore, our results are proved in the framework of GPTs 
after the quantum version of our results is stated in Section~\ref{summary}. 
Furthermore, due to the generality of our setting, 
our result can be used to investigate the structure of linear maps with the invariance of a cone, 
which is studied in Perron-Frobenius theory.

The remaining is organized as follows. 
Section~\ref{summary} is a summary of results in Section~\ref{asym} in the quantum case. 
As preparation for later discussion, 
Section~\ref{GPTs} describes the framework of GPTs. 
Section~\ref{d-maps} characterizes dynamical maps in the framework of GPTs. 
In this framework, Sections~\ref{ad}--\ref{mix} state our results 
on asymptotic properties for Markovian dynamics with discrete-time, namely, 
asymptotic decoupling, mixing, and ergodicity. 
Section~\ref{continuous} proceeds to the continuous case, 
and show the same results as the discrete case 
by using results of discrete-time evolution. 
Section~\ref{PF} gives a simple application of a result in Section~\ref{mix} 
to Perron-Frobenius theory. 
Section~\ref{conclusion} is our conclusion.

\section{Summary of our results in quantum theory}\label{summary}
We summarize results in Section~\ref{asym} in the quantum case. 
First, we remark on quantum channels which compose Markovian quantum dynamics. 
In quantum theory, quantum channels are given as TP-CP maps 
on the set $\cT(\cH)$ of all Hermitian matrices on a finite-dimensional quantum system $\cH$. 
The set $\cT(\cH)$ can be regarded as a finite-dimensional real Hilbert space 
equipped with the Hilbert-Schmidt inner product $\braket{X,Y}=\Tr XY$. 
In particular, quantum channels are given as linear maps on $\cT(\cH)$. 
Most results in this paper only require linearity and positivity, 
but for simplicity we focus on only TP-CP maps in this section.

Next, let us introduce some notations on a bipartite quantum system $\cH_1\otimes\cH_2$. 
Throughout this paper, we use \textit{tilde} to express to be bipartite. 
For instance, a bipartite quantum state is denoted by $\tilde{\rho}$, 
and a bipartite quantum channel is denoted by $\tilde{\Gamma}$. 
For a number $i\in\{1,2\}$ and a quantum state $\tilde{\rho}$ on $\cH_1\otimes\cH_2$, 
the reduced state on $\cH_i$ of $\tilde{\rho}$ is denoted by $\pi_i(\tilde{\rho})$. 
That is, by using the partial traces $\Tr_1$ and $\Tr_2$, we have 
$\pi_1(\tilde{\rho})=\Tr_2 \tilde{\rho}$ and $\pi_2(\tilde{\rho})=\Tr_1 \tilde{\rho}$.

\subsection{Discrete-time evolution}\label{q-discrete}
As Figure~\ref{F01}, let us consider dynamics 
when a quantum channel $\Gamma$ is applied to an initial quantum state $\rho$ many times. 
Then the dynamics is Markovian and discrete-time evolution. 
Now, in the bipartite case, 
we are interested in the asymptotic behavior of 
the $n$-th state $\tilde{\Gamma}^n(\tilde{\rho})$ as $n\to\infty$, 
and especially interested in whether the correlation between two quantum systems vanishes asymptotically. 
If the correlation vanishes asymptotically, 
we say that $\tilde{\Gamma}$ is asymptotically decoupling. 
Mathematically, asymptotic decoupling is defined as follows: 

\begin{figure}[t]
	\centering
	\includegraphics[scale=0.6]{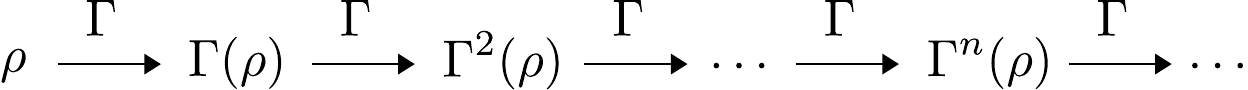}
	\caption{Markovian quantum dynamics with discrete-time.}
	\label{F01}
\end{figure}

\begin{definition}[Asymptotic decoupling]\label{def:q-ad}
	A quantum channel $\tilde{\Gamma}$ is asymptotically decoupling 
	if any state $\tilde{\rho}$ satisfies 
	\begin{equation*}
		\tilde{\Gamma}^n(\tilde{\rho})
		= \pi_1(\tilde{\Gamma}^n(\tilde{\rho}))\otimes\pi_2(\tilde{\Gamma}^n(\tilde{\rho})) + o(1)
		\quad(n\to\infty).
	\end{equation*}
\end{definition}

Since the above right-hand side is a product state, 
it means to be no correlation. 
Next, to clarify a necessary and sufficient condition of asymptotic decoupling, 
we introduce another asymptotic property, namely, mixing: 

\begin{definition}[Mixing]\label{def:q-mix}
	A quantum channel $\Gamma$ is mixing if 
	there exists a state $\rho_0$ 
	such that any state $\rho$ satisfies 
	\begin{equation}
		\lim_{n\to\infty} \Gamma^n(\rho) = \rho_0.
		\label{eq:mix}
	\end{equation}
	The state $\rho_0$ is called the stationary state.
\end{definition}

By using the above term, 
let us give a necessary and sufficient condition of asymptotic decoupling. 
For simplicity, first, we give it 
for a tensor product quantum channel $\Gamma_1\otimes\Gamma_2$: 

\begin{theorem}\label{thm:q-adD}
	For any two quantum channels $\Gamma_1$ and $\Gamma_2$, 
	the following conditions are equivalent.
	\begin{enumerate}
		\item\label{q-adD1}
		$\Gamma_1\otimes\Gamma_2$ is asymptotically decoupling.
		\item\label{q-adD2}
		$\Gamma_1$ or $\Gamma_2$ at least one is mixing.
	\end{enumerate}
\end{theorem}

Condition~\eqref{q-adD2} can be represented as a simple phrase \textit{local mixing}. 
Hence, simply speaking, asymptotic decoupling is equivalent to local mixing. 
Since spectral criterion \cite[Theorem~7]{BG} is known as a criterion of mixing, 
the computation of eigenequations determines whether a tensor product quantum channel is asymptotically decoupling or not. 
These are why Theorem~\ref{thm:q-adD} is simple and meaningful. 
Moreover, the above equivalence also holds for a general quantum channel: 

\begin{theorem}\label{thm:q-ad'D}
	For any quantum channel $\tilde{\Gamma}$, 
	the following conditions are equivalent.
	\begin{enumerate}
		\item\label{q-ad'D1}
		$\tilde{\Gamma}$ is asymptotically decoupling.
		\item\label{q-ad'D2}
		There exist a number $i_1\in\{1,2\}$ and a state $\rho_{0,i_1}$ on $\cH_{i_1}$ 
		such that any state $\tilde{\rho}$ satisfies 
		\[
		\tilde{\Gamma}^n(\tilde{\rho})
		= \rho_{0,i_1}\otimes\pi_{i_2}(\tilde{\Gamma}^n(\tilde{\rho})) + o(1)
		\quad (n\to\infty),
		\]
		where $i_2$ is an element of $\{1,2\}$ except for $i_1$.
	\end{enumerate}
\end{theorem}

Condition~\eqref{q-ad'D2} can be explicitly written as 
\begin{description}
	\item[Case $(i_1,i_2)=(1,2)$\hspace{1em}]
	$\tilde{\Gamma}^n(\tilde{\rho})
	= \rho_{0,1}\otimes\pi_2(\tilde{\Gamma}^n(\tilde{\rho})) + o(1)$,
	\item[Case $(i_1,i_2)=(2,1)$\hspace{1em}]
	$\tilde{\Gamma}^n(\tilde{\rho})
	= \pi_1(\tilde{\Gamma}^n(\tilde{\rho}))\otimes\rho_{0,2} + o(1)$.
\end{description}
Since the state $\rho_{0,1}$ or $\rho_{0,2}$ is something like a stationary state, 
condition~\eqref{q-ad'D2} can also be regarded as local mixing. 
Hence, for a general quantum channel, the same simple equivalence also holds: 
asymptotic decoupling is equivalent to local mixing.

For the above equivalence, 
it is meaningful to investigate a necessary and sufficient condition of mixing. 
Of course, as already mentioned, 
spectral criterion \cite[Theorem~7]{BG} is known as a criterion of mixing, 
but it requires us to solve an eigenequation. 
Hence, let us give another criterion of mixing 
that is a system of linear equations. 
For this purpose, 
we introduce another asymptotic property, namely, ergodicity: 

\begin{definition}[Ergodicity]\label{def:q-erg}
	A quantum channel $\Gamma$ is ergodic 
	if there exists a state $\rho_0$ such that any state $\rho$ satisfies 
	\begin{equation}
		\lim_{t\to\infty} \frac{1}{n}\sum_{k=0}^{n-1} \Gamma^k(\rho) = \rho_0.
		\label{eq:erg}
	\end{equation}
	The state $\rho_0$ is called the stationary state.
\end{definition}

Since the right-hand side of \eqref{eq:erg} is the Ces\`aro mean of the right-hand side of \eqref{eq:mix}, 
mixing implies ergodicity, but the converse does not necessarily hold. 
Definition~\ref{def:q-erg} cannot be checked directly by computing, 
but its numerical check is realized by the following proposition: 

\begin{proposition}\label{prop:q-ergD}
	For any quantum channel $\Gamma$, 
	the following conditions are equivalent.
	\begin{enumerate}
		\item\label{q-ergD1}
		$\Gamma$ is ergodic.
		\item\label{q-ergD2}
		$\dim\Ker(\Gamma - \id_{\cT(\cH)})=1$.
	\end{enumerate}
	The symbol $\id_{\cT(\cH)}$ denotes the identity map on $\cT(\cH)$.
\end{proposition}

Actually, as proved in Section~\ref{mix}, 
a quantum channel $\Gamma$ is mixing if and only if $\Gamma^{\otimes2}$ is ergodic. 
This equivalence and Proposition~\ref{prop:q-ergD} imply the following theorem 
which achieves the second goal of this paper: 

\begin{theorem}\label{thm:q-mixD}
	For any quantum channel $\Gamma$, 
	the following conditions are equivalent.
	\begin{enumerate}
		\item\label{q-mixD1}
		$\Gamma$ is mixing.
		\item\label{q-mixD2}
		$\Gamma^{\otimes2}$ is mixing.
		\item\label{q-mixD3}
		$\Gamma^{\otimes2}$ is ergodic.
		\item\label{q-mixD4}
		$\dim\Ker(\Gamma^{\otimes2} - \id_{\cT(\cH^{\otimes2})})=1$.
	\end{enumerate}
\end{theorem}

Proposition~\ref{prop:q-ergD} is well-known 
(for instance, see \cite[appendix]{BG}, \cite[Corollary~2]{Burgarth}), 
but Theorem~\ref{thm:q-mixD} is not fully known as far as we know, 
and only partial results are published. 
For instance, the equivalence of conditions~\eqref{q-mixD1} and \eqref{q-mixD3} in Theorem~\ref{thm:q-mixD} 
was proved for unital quantum channels in finite-dimensions \cite[Theorem~2.10]{JR}. 
An equivalence of conditions similar to conditions~\eqref{q-mixD1} and \eqref{q-mixD3} 
is known for unital normal CP maps in infinite-dimensions \cite[Theorem~6.3]{Watanabe}. 
However, only a few preceding studies considered tensor product channels in the first place \cite{JR,Watanabe, Luczak}. 
Their proofs are of operator algebra, 
but our proof is of linear algebra, 
and thus Theorem~\ref{thm:q-mixD} also holds in GPTs. 

Finally, summarizing Theorems~\ref{thm:q-adD} and \ref{thm:q-mixD}, 
we have the following theorem immediately: 

\begin{theorem}\label{thm:q-mix-ad}
	Let $\Gamma_1$ and $\Gamma_2$ be quantum channels. 
	If $\Gamma_2$ is not mixing, then the following conditions are equivalent.
	\begin{enumerate}
		\item\label{q-mix-ad1}
		$\Gamma_1$ is mixing.
		\item\label{q-mix-ad2}
		$\Gamma_1^{\otimes2}$ is mixing.
		\item\label{q-mix-ad3}
		$\Gamma_1^{\otimes2}$ is ergodic.
		\item\label{q-mix-ad4}
		$\dim\Ker(\Gamma_1^{\otimes2} - \id_{\cT(\cH_1^{\otimes2})})=1$.
		\item\label{q-mix-ad5}
		$\Gamma_1\otimes\Gamma_2$ is asymptotically decoupling.
	\end{enumerate}
\end{theorem}

\subsection{Continuous-time evolution}\label{q-continuous}
Next, we address continuous-time evolution. 
In the discrete case, 
$\Gamma^n(\rho)$ has denoted the state at time $n\in\mathbb{N}$. 
Since we address continuous-time evolution in this subsection, 
we denote by $\Gamma^{(t)}(\rho)$ the state at time $t>0$. 
Moreover, it is natural that the family $\{\Gamma^{(t)}\}_{t>0}$ of quantum channels 
should satisfy 
\begin{itemize}
	\item
	$\Gamma^{(t)} \circ \Gamma^{(s)} = \Gamma^{(t+s)}$ for all $t,s>0$,
	\item
	$\Gamma^{(t)} \to \id_{\cT(\cH)}$ as $t\downarrow0$.
\end{itemize}
The above first condition is illustrated by Figure~\ref{F02}. 
It asserts that the state at time $t+s$ equals 
the state after a time $s$ passes from the state at time $t$. 
A family $\{\Gamma^{(t)}\}_{t>0}$ satisfying the above two conditions 
is called a \textit{$C_0$ semigroup}. 
It can be easily checked that 
any $C_0$ semigroup $\{\Gamma^{(t)}\}_{t>0}$ is right-continuous everywhere. 
The definition of a $C_0$ semigroup is simple and intuitive in the finite-dimensional case, 
but the infinite-dimensional case needs a few technical conditions. 
Although there are many studies on $C_0$ semigroups of finite/infinite-dimensions, 
we use no existing results on $C_0$ semigroups and only use the above two conditions. 
In the context of Markovian quantum dynamics, 
$C_0$ semigroups are also called \textit{quantum dynamical semigroups} \cite{Lindblad}.

\begin{figure}[t]
	\centering
	\includegraphics[scale=0.6]{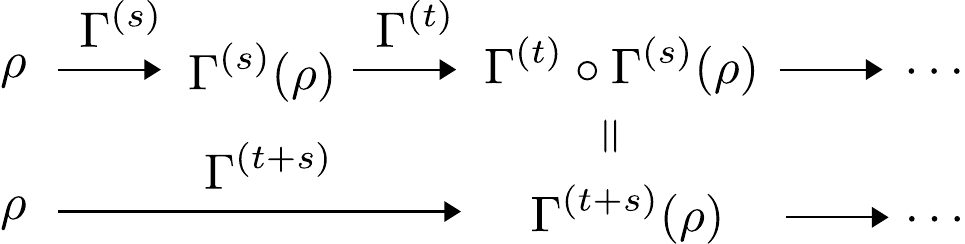}
	\caption{Markovian quantum dynamics with continuous-time.}
	\label{F02}
\end{figure}

In order to state our results in the continuous case, 
we need to define the continuous versions of asymptotic decoupling, mixing, and ergodicity. 
Fortunately, once replacing $\Gamma^n$ and $n\to\infty$ with $\Gamma^{(t)}$ and $t\to\infty$ respectively, 
the continuous versions of asymptotic decoupling and mixing are defined. 
However, to define the continuous version of ergodicity, 
we need to use an integral as follows: 

\begin{definition}[Ergodicity]
	A $C_0$ semigroup $\{\Gamma^{(t)}\}_{t>0}$ is ergodic 
	if there exists a state $\rho_0$ such that any state $\rho$ satisfies 
	\begin{equation*}
		\lim_{t\to\infty} \frac{1}{t}\int_0^t \Gamma^{(s)}(\rho)\,ds = \rho_0.
	\end{equation*}
	The state $\rho_0$ is called the stationary state.
\end{definition}

Any mixing $C_0$ semigroup $\{\Gamma^{(t)}\}_{t>0}$ is ergodic 
in the same way as the discrete case. 
This fact is due to L'Hospital's rule. 
In the above setting, our results are below: 

\begin{theorem}\label{thm:q-c-adD}
	For any two $C_0$ semigroups $\{\Gamma_1^{(t)}\}_{t>0}$ and $\{\Gamma_2^{(t)}\}_{t>0}$, 
	the following conditions are equivalent.
	\begin{enumerate}
		\item\label{q-c-adD1}
		$\{\Gamma_1^{(t)}\otimes\Gamma_2^{(t)}\}_{t>0}$ is asymptotically decoupling.
		\item\label{q-c-adD2}
		$\{\Gamma_1^{(t)}\}_{t>0}$ or $\{\Gamma_2^{(t)}\}_{t>0}$ at least one is mixing.
	\end{enumerate}
\end{theorem}

\begin{theorem}\label{thm:q-c-ad'D}
	For any $C_0$ semigroup $\{\tilde{\Gamma}^{(t)}\}_{t>0}$, 
	the following conditions are equivalent.
	\begin{enumerate}
		\item\label{q-c-ad'D1}
		$\{\tilde{\Gamma}^{(t)}\}_{t>0}$ is asymptotically decoupling.
		\item\label{q-c-ad'D2}
		There exist a number $i_1\in\{1,2\}$ and a state $\rho_{0,i_1}$ on $\cH_{i_1}$ 
		such that any state $\tilde{\rho}$ satisfies 
		\[
		\tilde{\Gamma}^{(t)}(\tilde{\rho})
		= \rho_{0,i_1}\otimes\pi_{i_2}(\tilde{\Gamma}^{(t)}(\tilde{\rho})) + o(1)
		\quad (t\to\infty),
		\]
		where $i_2$ is an element of $\{1,2\}$ except for $i_1$.
	\end{enumerate}
\end{theorem}

\begin{theorem}\label{thm:q-c-mixD}
	For any $C_0$ semigroup $\{\Gamma^{(t)}\}_{t>0}$, 
	the following conditions are equivalent.
	\begin{enumerate}
		\item\label{q-c-mixD1}
		$\{\Gamma^{(t)}\}_{t>0}$ is mixing.
		\item\label{q-c-mixD2}
		$\{(\Gamma^{(t)})^{\otimes2}\}_{t>0}$ is mixing.
		\item\label{q-c-mixD3}
		$\{(\Gamma^{(t)})^{\otimes2}\}_{t>0}$ is ergodic.
		\item\label{q-c-mixD4}
		$\Gamma^{(\epsilon)}$ is mixing for some $\epsilon>0$.
		\item\label{q-c-mixD5}
		$(\Gamma^{(\epsilon)})^{\otimes2}$ is mixing for some $\epsilon>0$.
		\item\label{q-c-mixD6}
		$(\Gamma^{(\epsilon)})^{\otimes2}$ is ergodic for some $\epsilon>0$.
	\end{enumerate}
\end{theorem}

Similarly to the discrete case, 
asymptotic decoupling is equivalent to local mixing, 
and mixing for dynamics is equivalent to ergodicity for the two-fold tensor product of dynamics.

Some readers might consider that 
Theorems~\ref{thm:q-c-adD}, \ref{thm:q-c-ad'D}, and \ref{thm:q-c-mixD} are more general than 
Theorems~\ref{thm:q-adD}, \ref{thm:q-ad'D}, and \ref{thm:q-mixD}. 
However, it is not true due to the following reason. 
Any $C_0$ semigroup $\{\Gamma^{(t)}\}_{t>0}$ can be represented as an exponential function \cite[Theorem~I.3.7]{EN}: 
there exists a linear map $\cL$ such that 
$\Gamma^{(t)} = e^{t\cL}$ for all $t>0$. 
Thus $\Gamma^{(t)}$ does not have the eigenvalue zero. 
However, the discrete case allows that $\Gamma$ has the eigenvalue zero. 
Moreover, 
$\{\Gamma^{(t)}\}_{t>0}$ is a $C_0$ semigroup of CP maps 
if and only if $\{(\Gamma^{(t)})^{\otimes2}\}_{t>0}$ is a $C_0$ semigroup of positive maps \cite[Theorem~1]{BCF}. 
This fact is completely different from the case with CP maps: 
positivity for $\Gamma^{\otimes2}$ does not imply complete positivity for $\Gamma$. 
These are why the continuous case is somewhat restricted. 
The above generator $\cL$ is called a \textit{Lindbladian} 
and was characterized by Lindblad \cite{Lindblad}.

\section{GPTs with cones}\label{GPTs}
On the basis of GPTs, let us define \textit{states} and \textit{measurements}. 
GPTs are a general framework including classical probability theory and quantum theory \cite{Short-Wehner,Janotta-Hinrichsen,Lami}. 
Simply speaking, a GPT consists of a \textit{proper cone} $\cK$ and a \textit{unit effect} $u$. 
This is because the set $\cS(\cK,u)$ of all states is defined by using $\cK$ and $u$. 
Throughout this paper, we consider finite-dimensional GPTs alone. 
Also, we summarize basic lemmas on convex cones used implicitly/explicitly in this paper, in Appendix~\ref{appC}.

First, we provide preliminary knowledge below. 
We denote by $\cV^\ast$ the dual space of a finite-dimensional real vector space $\cV$, 
and assume $\dim\cV\ge2$ throughout this paper. 
The dual space $\cV^\ast$ can be naturally identified with $\cV$ 
by using an inner product $\braket{\cdot,\cdot}$ on $\cV$. 
Hence we identify $\cV^\ast$ with $\cV$. 
A convex set $\cK\subset\cV$ is called a \textit{convex cone} 
if $\cK$ contains any non-negative number multiple of any $x\in\cK$. 
When a convex cone $\cK$ is a closed set, 
we say that $\cK$ is a \textit{closed convex cone} (for short, \textit{cone}). 
For any non-empty convex cone $\cK\subset\cV$, 
the \textit{dual cone} is defined as 
\[
\cK^\ast = \set{y\in\cV | \braket{y,x}\ge0\ (\forall x\in\cK)}.
\]
It is well-known \cite{Boyd,Ekeland} that 
\begin{itemize}
	\item
	$\cK^\ast$ is a cone,
	\item
	$(\cK')^\ast\subset\cK^\ast$ if $\cK\subset\cK'$,
	\item
	$\cK^\ast=\cl(\cK)^\ast$,
	\item
	$\cK^{\ast\ast}=\cl(\cK)$,
\end{itemize}
where $\cl(\cX)$ denotes the closure of a subset $\cX\subset\cV$. 
In particular, the first three properties follow from the definition immediately. 
If $\cK^{\ast}=\cK$, the cone $\cK$ is called \textit{self-dual}. 
The norm on $\cV$ is defined as $\|x\| = \sqrt{\braket{x,x}}$ for all $x\in\cV$.

Figure~\ref{Fig1} is an instance of a cone of $\mathbb{R}^2$ and its dual cone. 
The boundary of the cone in the left figure consists of the two rays $l_1$ and $l_2$. 
On the other hand, the boundary of the dual cone in the right figure consists of the two rays $l_1^\bot$ and $l_2^\bot$. 
Hence, the larger the angle between $l_1$ and $l_2$ is, the smaller the angle between $l_1^\bot$ and $l_2^\bot$ is.

\begin{figure}[t]
	\centering
	\includegraphics[scale=1]{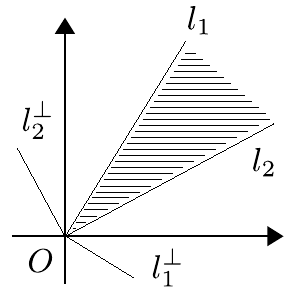}
	\hspace{2ex}
	\includegraphics[scale=1]{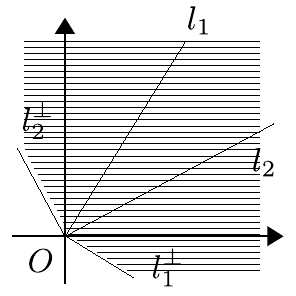}
	\caption{The left and right figures are a cone of $\mathbb{R}^2$ and its dual cone, respectively. 
	For each $i=1,2$, the rays $l_i$ and $l_i^\bot$ are orthogonal to each other at the origin $O$.}\label{Fig1}
\end{figure}

Next, we describe a GPT, which requires the following components. 
Suppose that $\cK\subset\cV$ is a \textit{proper cone}, i.e., 
a cone satisfying $\tint(\cK) \not= \emptyset$ and $\cK\cap(-\cK) = \{ 0 \}$, 
where $\tint(\cX)$ denotes the interior of a subset $\cX\subset\cV$. 
Once we fix a \textit{unit effect} $u\in\tint(\cK^\ast)$, 
the set of all \textit{states} is given as 
\[
\cS(\cK,u) \coloneqq \set{x \in \cK | \braket{u,x}=1}.
\]
The set $\cS(\cK,u)$ is a compact convex set 
(see Lemma~\ref{state-compact} in Appendix~\ref{appC}). 
A \textit{measurement} is given as a decomposition $\{e_i\}_i$ of $u$, namely, 
it satisfies $e_i\in\cK^\ast$ and $\sum_i e_i=u$, where each $i$ corresponds to an outcome. 
When $x$ is a state and a measurement $\{e_i\}_i$ is performed, 
the probability to obtain an outcome $i$ is $\braket{e_i,x}$. 
Therefore, once we fix the tuple $(\cV, \braket{\cdot,\cdot}, \cK, u)$, 
our GPT is established.

Next, we give two typical examples of GPTs. 
Table~\ref{T01} summarizes the tuples $(\cV, \braket{\cdot,\cdot}, \cK, u)$ which appear below.

\begin{example}[Classical probability theory]\label{CPT}
In order to recover classical probability theory with $d$ outcomes, 
put $\cV=\mathbb{R}^d$ and $\cK=[0,\infty)^d$. 
Also, choose the inner product $\braket{\cdot,\cdot}$ to be the standard inner product. 
Hence the relation $\cK^\ast=[0,\infty)^d$ holds, 
and thus $\cK$ is self-dual. 
Furthermore, choosing the unit effect $u=[1,\ldots,1]^\top$, 
we find that states equal probability vectors. 
Thus we obtain classical probability theory with $d$ outcomes.
\end{example}

\begin{example}[Quantum theory]\label{QT}
Next, let us recover quantum theory on a finite-dimensional complex Hilbert space $\cH$. 
Choose $\cV$ to be the set $\cT(\cH)$ of all Hermitian matrices on $\cH$. 
Also, choose $\cK$ to be the set $\cT_+(\cH)$ of all positive semi-definite matrices, 
which has non-empty interior and satisfies $\cT_+(\cH)\cap(-\cT_+(\cH))=\{ 0 \}$. 
Furthermore, define the inner product $\braket{\cdot,\cdot}$ as $\braket{Y,X}=\Tr YX$ for all $X,Y\in\cT(\cH)$. 
Hence the relation $\cK^\ast=\cT_+(\cH)$ holds, 
and thus $\cK$ is self-dual. 
Choosing the unit effect $u$ to be the identity matrix on $\cH$, 
we find that states equal density matrices. 
Finally, we note that measurements are given as positive operator-valued measures (POVMs).
\end{example}

\begin{table}[t]
	\centering
	\renewcommand{\arraystretch}{1.2}
	\caption{The two tuples $(\cV, \braket{\cdot,\cdot}, \cK, u)$ in Examples~\ref{CPT} and \ref{QT}.}
	\label{T01}
	\begin{tabular}{|c|cc|}
		\hline
		   &Vector space $\cV$&Inner product $\braket{\cdot,\cdot}$\\
		\hline\hline
		Classical&$\mathbb{R}^d$&$\sum_{i=1}^d y_i x_i$\\
		Quantum&All Hermitian matrices&$\Tr YX$\\
		\hline\hline
		   &Proper cone $\cK$&Unit effect $u$\\
		\hline\hline
		Classical&$[0,\infty)^d$&$[1,\ldots,1]^\top$\\
		Quantum&All positive semi-definite matrices&Identity matrix\\
		\hline
	\end{tabular}
\end{table}

Now, we return to GPTs and focus on two GPTs given 
by two tuples $(\cV_i, \braket{\cdot,\cdot}, \cK_i, u_i)$ with $i=1,2$. 
The \textit{joint system} is given as the tensor product space $\cV_1\otimes\cV_2$ 
equipped with the natural inner product induced by the inner products of $\cV_1$ and $\cV_2$. 
Then the \textit{joint unit effect} is $u_1\otimes u_2$. 
When two states $x_1\in\cS(\cK_1,u_1)$ and $x_2\in\cS(\cK_2,u_2)$ 
are prepared independently in the respective systems, 
it is natural that the state on the joint system is given as the product state $x_1\otimes x_2$. 
Since any convex combination of product states is also realized by randomization, 
any state $x\in\cS(\cK_1\otimes\cK_2,u_1\otimes u_2)$ can be realized. 
Here we define the tensor product cone $\cK_1\otimes\cK_2$ as 
\[
\cK_1\otimes\cK_2
= \Set{\sum_{i=1}^n x_{i,1}\otimes x_{i,2}\in\cV_1\otimes\cV_2 | 
n\in\mathbb{N},\ x_{i,1}\in\cK_1,\ x_{i,2}\in\cK_2\ (1\le i\le n)}.
\]
Moreover, any $\sum_{i=1}^m x_{i,1}\otimes x_{i,2}\in\cK_1\otimes\cK_2$ 
and any $\sum_{j=1}^n y_{j,1}\otimes y_{j,2}\in\cK_1^\ast\otimes\cK_2^\ast$ satisfy 
\[
\Bigl\langle \sum_{i=1}^m x_{i,1}\otimes x_{i,2},\ \sum_{j=1}^n y_{j,1}\otimes y_{j,2} \Bigr\rangle
= \sum_{i=1}^m \sum_{j=1}^n \braket{x_{i,1}, y_{j,1}}\braket{x_{i,2}, y_{j,2}} \ge 0.
\]
Thus the following proposition holds: 

\begin{proposition}\label{prop:cone}
	If $\cK_1 $ and $\cK_2$ are cones, then 
	\[
	\cK_1^\ast\otimes\cK_2^\ast \subset (\cK_1\otimes\cK_2)^\ast.
	\]
\end{proposition}

Finally, we consider what properties a cone $\tilde{\cK}$ of the joint system should satisfy. 
From the argument in the previous paragraph, the cone $\tilde{\cK}$ must contain $\cK_1\otimes\cK_2$. 
Similarly, by considering that two measurements are performed independently in the respective systems, 
we obtain $\cK_1^\ast\otimes\cK_2^\ast\subset\tilde{\cK}^\ast$. 
Hence $\tilde{\cK}$ must satisfy 
\[
\tilde{\cK}_{\min} \coloneqq \cK_1\otimes\cK_2
\subset \tilde{\cK}
\subset (\cK_1^\ast\otimes\cK_2^\ast)^\ast \eqqcolon \tilde{\cK}_{\max},
\]
where we note that the relation $\tilde{\cK}_{\min}\subset\tilde{\cK}_{\max}$ holds due to Proposition~\ref{prop:cone}. 
Of course, the cone $\tilde{\cK}$ is not necessarily unique, 
since $\tilde{\cK}_{\min}\not=\tilde{\cK}_{\max}$ in general. 
In fact, even if both $\cK_1$ and $\cK_2$ are self-dual, 
the cone $\tilde{\cK}_{\min}$ is not necessarily self-dual. 
To see this fact, we discuss two quantum systems as follows: 

\begin{example}
We consider quantum theory on finite-dimensional complex Hilbert spaces $\cH_1$ and $\cH_2$. 
Then the vector spaces of the first, second, and joint systems are 
$\cT(\cH_1)$, $\cT(\cH_2)$, and $\cT(\cH_1\otimes\cH_2)$, respectively. 
States on the joint system are density matrices in $\cT_+(\cH_1\otimes\cH_2)$. 
Let $\ket{a_i},\ket{b_i}\in\cH_i$ be orthonormal vectors for each $i=1,2$. 
When defining the unit vector $\ket{\Phi} = (\ket{a_1}\ket{a_2} + \ket{b_1}\ket{b_2})/\sqrt{2}\in\cH_1\otimes\cH_2$, 
it is well-known that the state $\ketbra{\Phi}{\Phi}$ on the joint system is not \textit{separable} \cite[Section~1.4]{Hayashi}, i.e., 
\[
\ketbra{\Phi}{\Phi} \not\in \Sep(\cH_1;\cH_2) \coloneqq \cT_+(\cH_1)\otimes\cT_+(\cH_2).
\]
Thus the cone $\cT_+(\cH_1\otimes\cH_2)$ strictly contains $\Sep(\cH_1;\cH_2)$. 
Noting the inclusion relation 
\[
\Sep(\cH_1;\cH_2) \subset \cT_+(\cH_1 \otimes  \cH_2)
\subset \Sep(\cH_1;\cH_2)^\ast,
\]
we obtain $\Sep(\cH_1;\cH_2)\not=\Sep(\cH_1;\cH_2)^\ast$. 
Therefore, the cones $\cT_+(\cH_1)$ and $\cT_+(\cH_2)$ are self-dual, 
but $\Sep(\cH_1;\cH_2)$ is not self-dual.
\end{example}

\section{$\cK$-positive maps and $(\cK,u)$-dynamical maps}\label{d-maps}
In order to address Markovian dynamics, 
we describe \textit{dynamical maps}. 
Suppose that two tuples $(\cV, \braket{\cdot,\cdot}, \cK, u)$ 
and $(\cV', \braket{\cdot,\cdot}, \cK', u')$ give GPTs. 
Since a dynamical map $A$ transmits states $x\in\cS(\cK,u)$ to states $Ax\in\cS(\cK',u')$, 
a dynamical map must satisfy some properties. 
First, a dynamical map $A$ is a linear map from $\cV$ into $\cV'$ 
because $A$ needs to preserve the convex combination structure. 
Moreover, the linear map $A$ must satisfy the following properties: 

\begin{definition}[$(\cK; \cK')$-positivity]
	If a linear map $A: \cV\to\cV'$ satisfies $A\cK\subset\cK'$, 
	we say that $A$ is $(\cK; \cK')$-positive.
\end{definition}

\begin{definition}[Dual $(u; u')$-preserving]
	If a linear map $A: \cV\to\cV'$ satisfies $A^\ast u'=u$, 
	we say that $A$ is dual $(u; u')$-preserving, 
	where $A^\ast$ denotes the adjoint map of $A$.
\end{definition}

If a linear map $A$ is $(\cK; \cK')$-positive and dual $(u; u')$-preserving, 
we say that $A$ is \textit{$(\cK,u; \cK',u')$-dynamical}. 
A linear map $A$ is $(\cK,u; \cK',u')$-dynamical if and only if 
the inclusion relation $A\cS(\cK,u)\subset\cS(\cK',u')$ holds. 
In this paper, 
we consider the case 
$(\cV, \braket{\cdot,\cdot}, \cK, u) = (\cV', \braket{\cdot,\cdot}, \cK', u')$ 
unless otherwise noted. 
In this case, the above properties are called 
$\cK$-positivity, dual $u$-preserving, and $(\cK,u)$-dynamical property, respectively. 
Unless there is confusion, 
a $(\cK,u)$-dynamical map is called a \textit{dynamical map} simply.

The definition of dynamical maps is different from that of quantum channels. 
In quantum theory, quantum channels are given as TP-CP maps. 
Trace-preserving and positivity in quantum theory correspond to dual $u$-preserving and $\cK$-positivity, respectively. 
Hence a $(\cK,u)$-dynamical map is a counterpart of a trace-preserving and positive map in quantum theory. 
Moreover, our class is strictly lager than the class of TP-CP maps. 
The correspondences are summarized in Table~\ref{T1}.

\begin{table}[t]
	\centering
	\caption{Properties for linear maps in quantum theory and GPTs.}\label{T1}
	\vspace{2ex}
	\begin{tabular}{cc}
		Quantum theory&GPTs\\
		\hline
		Positivity&$\cK$-positivity\\
		Trace-preserving (TP)&Dual $u$-preserving\\
		Complete positivity (CP)&   \\
		TP and positivity&$(\cK,u)$-dynamical property
	\end{tabular}
\end{table}

As stated in Section~\ref{introduction}, 
to handle stochastic operation including SLOCC, 
we need to address dynamical maps without dual $u$-preserving. 
Fortunately, most results in this paper do not require dual $u$-preserving. 
Hence we focus on $\cK$-positive maps mainly. 
As known in preceding studies \cite{KR,Vander}, 
$\cK$-positivity guarantees a few good properties. 
Before stating them, we introduce several basic terms in linear algebra below.

For a linear map $A$ and its eigenvalue $\lambda$, 
the multiplicity as a root of the characteristic polynomial is called the \textit{algebraic multiplicity} 
and denoted by $\am(A; \lambda)$. 
Also, the dimension of the eigenspace $\cW(A; \lambda) \coloneqq \Ker(A - \lambda\id_{\cV})$ is called the \textit{geometric multiplicity} 
and denoted by $\gm(A; \lambda)$. 
For convenience, we define $\am(A; \lambda)=\gm(A; \lambda)=0$ 
if $\lambda$ is not an eigenvalue of $A$. 
For any linear map $A$ and any $\lambda\in\mathbb{C}$, 
the inequality $\gm(A; \lambda) \le \am(A; \lambda)$ holds. 
As another basic term, 
the greatest absolute value of all eigenvalues of a linear map $A$ 
is called the \textit{spectral radius} and denoted by $r(A)$ \cite{Wolf}. 
The spectral radius satisfies 
\[
r(A) = \limsup_{n\to\infty} \|A^n\|^{1/n},
\]
where $\|A\|$ is the operator norm of $A$ based on the norm on $\cV$. 
Since two arbitrary norms on a finite-dimensional vector space are uniformly equivalent to each other, 
we can select any norm instead of that above. 
Moreover, the equations $r(A_1\otimes A_2)=r(A_1)r(A_2)$ and $r(A^\ast)=r(A)$ hold. 
For details, see \cite{Wolf}.

In addition to the above terms, 
we also define the \textit{degree} of an eigenvalue of a linear map: 

\begin{definition}[Degree of eigenvalues {\cite[Definition~3.1]{Vander}}]
	For a linear map $A$ and its eigenvalue $\lambda$, 
	the degree of $\lambda$ is defined as the size of the largest Jordan block associated with $\lambda$, 
	in the Jordan canonical form of $A$.
\end{definition}

For instance, the matrix 
\[
\begin{bmatrix}
	0\\
	 &0&1\\
	 &0&0\\
	 & & &1\\
	 & & & &1
\end{bmatrix},
\]
which is already a Jordan canonical form, has the eigenvalues $0$ and $1$. 
In this case, the degrees of $0$ and $1$ are two and one, respectively. 
Now, we have been ready to state the following proposition \cite[Theorem~3.1]{Vander}: 

\begin{proposition}\label{prop:pos}
	If $A$ is a $\cK$-positive map, then the following properties hold.
	\begin{itemize}
		\item
		$r(A)$ is an eigenvalue of $A$.
		\item
		$\cK$ contains an eigenvector of $A$ associated with $r(A)$.
		\item
		The degree of $r(A)$ is greater than or equal to 
		the degree of any other eigenvalue whose absolute value is $r(A)$.
	\end{itemize}
\end{proposition}

If $\am(A; r(A))=1$, 
then the Jordan block associated with $r(A)$ is the matrix $[r(A)]$ alone, 
and thus the degree of $r(A)$ is one. 
In this case, Proposition~\ref{prop:pos} implies that 
any Jordan block associated with any eigenvalue $\lambda$ satisfying $|\lambda|=r(A)$ equals the matrix $[\lambda]$. 
In other words, $\am(A; \lambda)=\gm(A; \lambda)$ for any eigenvalue $\lambda$ satisfying $|\lambda|=r(A)$. 
We summarize the above argument as the following corollary: 

\begin{corollary}\label{coro:pos}
	If $A$ is a $\cK$-positive map and satisfies $\am(A; r(A))=1$, 
	then $\am(A; \lambda)=\gm(A; \lambda)$ for any eigenvalue $\lambda$ whose absolute value is $r(A)$.
\end{corollary}

Finally, we show the following basic propositions. 

\begin{proposition}\label{dual-pos}
	If a linear map $A$ is $\cK$-positive, 
	then $A^\ast$ is $\cK^\ast$-positive.
\end{proposition}
\begin{proof}
	Let $y\in\cK^\ast$. 
	Any $x\in\cK$ satisfies $\braket{A^\ast y,x}=\braket{y,Ax}\ge0$, 
	whence $A^\ast y\in\cK^\ast$.
\end{proof}

\begin{proposition}\label{d-map}
	Any dynamical map $A$ satisfies $r(A)=1$.
\end{proposition}
\begin{proof}
	Proposition~\ref{prop:pos} implies that $A$ has an eigenvector $x_0\in\cK$ associated with $r(A)$. 
	Since $A^\ast u=u$ and $u\in\tint(\cK^\ast)$, we have 
	$\braket{u,x_0} = \braket{A^\ast u,x_0} = \braket{u,Ax_0} = r(A)\braket{u,x_0}$ 
	and $\braket{u,x_0}>0$ because of Lemma~\ref{lem:appC1} in Appendix~\ref{appC}. 
	Therefore, $r(A)=1$.
\end{proof}

\section{Asymptotic properties for Markovian dynamics}\label{asym}
In this section, we address asymptotic properties for Markovian dynamics in GPTs. 
Sections~\ref{ad}--\ref{mix} is the discrete case, 
and Section~\ref{continuous} is the continuous case. 
Theorems~\ref{thm:q-adD}, \ref{thm:q-ad'D}, and \ref{thm:q-mixD} 
are derived from 
Theorems~\ref{thm:adD}, \ref{thm:ad'D}, and \ref{thm:mixP}, respectively. 
Also, Theorems~\ref{thm:q-c-adD}, \ref{thm:q-c-ad'D}, and \ref{thm:q-c-mixD} 
are derived from 
Theorems~\ref{thm:c-adD}, \ref{thm:c-ad'D}, and \ref{thm:c-mixD}, respectively. 
Throughout this section, 
suppose that a tuple $(\cV_i, \braket{\cdot,\cdot}, \cK_i, u_i)$ gives a GPT for each $i=1,2$. 
Also, suppose that a tuple $(\cV_1\otimes\cV_2, \braket{\cdot,\cdot}, \tilde{\cK}, u_1\otimes u_2)$ 
gives a GPT on the joint system, where 
$\tilde{\cK}_{\min} \subset \tilde{\cK} \subset \tilde{\cK}_{\max}$. 
We use \textit{tilde} to express to be bipartite 
in the same way as Section~\ref{summary}.

\subsection{Asymptotic decoupling}\label{ad}
First, to define asymptotic decoupling in the framework of GPTs, 
for a state $\tilde{x}$ on the joint system, 
we must define the \textit{reduced state} $\pi_i\tilde{x}$ on the $i$-th system. 
The reduced state $\pi_i\tilde{x}$ represents the state from a viewpoint of an observer of the $i$-th system. 
In our context, it is convenient to define the reduced states $\pi_1\tilde{x}$ and $\pi_2\tilde{x}$ 
for a state $\tilde{x}\in\cS(\tilde{\cK}_{\max}, u_1\otimes u_2)$ 
because the cone $\tilde{\cK}_{\max}$ is largest 
of all cones $\tilde{\cK}$ of the joint system. 
The \textit{reduced state} $\pi_1(\tilde{x}; u_1,u_2)$ of 
$\tilde{x}\in\cS(\tilde{\cK}_{\max}, u_1\otimes u_2)$ 
is defined by the condition 
\begin{equation}
	\forall y_1\in\cK_1^\ast,\quad
	\braket{y_1, \pi_1(\tilde{x}; u_1,u_2)} = \braket{y_1\otimes u_2, \tilde{x}}.
	\label{eq:pi}
\end{equation}
Since $\cK_1^\ast$ has non-empty interior, 
the Riesz representation theorem guarantees the unique existence of $\pi_1(\tilde{x}; u_1,u_2)\in\cV_1$. 
Let us verify $\pi_1(\tilde{x}; u_1,u_2)\in\cS(\cK_1,u_1)$ as follows. 
Since $\tilde{x}\in\cS(\tilde{\cK}_{\max}, u_1\otimes u_2)$, 
the right-hand side of \eqref{eq:pi} is greater than or equal to zero, 
which implies $\pi_1(\tilde{x}; u_1,u_2)\in\cK_1$. 
Moreover, putting $y_1=u_1$ in \eqref{eq:pi}, 
we find that 
\[
\braket{u_1, \pi_1(\tilde{x}; u_1,u_2)} = \braket{u_1\otimes u_2, \tilde{x}} = 1.
\]
Therefore, $\pi_1(\tilde{x}; u_1,u_2)\in\cS(\cK_1,u_1)$. 
The reduced state $\pi_2(\tilde{x}; u_1,u_2)\in\cS(\cK_2,u_2)$ is also defined similarly. 
Unless there is confusion, 
we express $\pi_i(\tilde{x}; u_1,u_2)$ as $\pi_i\tilde{x}$ simply. 
Then $\pi_i$ can be naturally regarded as a linear map from $\cV_1\otimes\cV_2$ onto $\cV_i$.

Next, to define asymptotic decoupling for (not necessarily $\tilde{\cK}$-positive) linear maps, 
we introduce the set 
\[
\cD_n(\tilde{A}) \coloneqq \Set{\tilde{x}\in\tilde{\cK}\setminus\{ 0 \} | 
\tilde{A}^k \tilde{x}
\in \tilde{\cK}_{\max}\setminus\{ 0 \}\ (\forall k\ge n)}
\]
for a linear map $\tilde{A}$ and $n\in\mathbb{N}$. 
Then the monotonicity 
$\cD_1(\tilde{A}) \subset \cD_2(\tilde{A}) \subset \cdots$ holds. 
By using the sets $\cD_n(\tilde{A})$, 
asymptotic decoupling is defined as follows: 

\begin{definition}[$(\tilde{\cK}, u_1, u_2)$-asymptotic decoupling]\label{def:adD}
	Suppose that a dual $u_1\otimes u_2$-preserving map $\tilde{A}$ 
	satisfies 
	\begin{equation}
		\bigcup_{n=1}^\infty \cD_n(\tilde{A}) = \tilde{\cK}\setminus\{ 0 \}.
		\label{ad-assumption}
	\end{equation}
	Then $\tilde{A}$ is $(\tilde{\cK}, u_1, u_2)$-asymptotically decoupling 
	if any state $\tilde{x}\in\cS(\tilde{\cK}, u_1\otimes u_2)$ satisfies 
	\begin{equation}
		\tilde{A}^n \tilde{x}
		= \pi_1\tilde{A}^n \tilde{x}\otimes\pi_2\tilde{A}^n \tilde{x} + o(1)
		\quad(n\to\infty). \label{eq:ad}
	\end{equation}
\end{definition}

The assumption \eqref{ad-assumption} allows us to consider the limit \eqref{eq:ad}. 
Definition~\ref{def:adD} needs \eqref{ad-assumption}, 
but one may not mind \eqref{ad-assumption} 
because for any two dynamical maps $A_1$ and $A_2$, 
the tensor product map $\tilde{A}=A_1\otimes A_2$ satisfies \eqref{ad-assumption} 
(see Lemma~\ref{lem:A1} in Appendix~\ref{appA}). 
In this case, we can simplify \eqref{eq:ad}: 
\[
(A_1\otimes A_2)^n(\tilde{x} - \pi_1\tilde{x}\otimes\pi_2\tilde{x}) = o(1).
\]

Next, to clarify a necessary and sufficient condition of asymptotic decoupling, 
we introduce another asymptotic property, namely, mixing: 

\begin{definition}[Mixing]\label{def:mixD}
	A dynamical map $A$ is mixing if 
	there exists a state $x_0\in\cS(\cK,u)$ such that 
	any state $x\in\cS(\cK,u)$ satisfies 
	\begin{equation*}
		\lim_{n\to\infty} A^n x = x_0.
	\end{equation*}
	The state $x_0$ is called the stationary state.
\end{definition}

In the same way as Section~\ref{summary}, 
we give a necessary and sufficient condition of asymptotic decoupling 
for a tensor product map $A_1\otimes A_2$ and a linear map $\tilde{A}$ in this order: 

\begin{theorem}\label{thm:adD}
	For any two dynamical maps $A_1$ and $A_2$, 
	the following conditions are equivalent.
	\begin{enumerate}
		\item\label{adD1}
		$A_1\otimes A_2$ is $(\tilde{\cK},u_1,u_2)$-asymptotically decoupling.
		\item\label{adD2}
		$A_1$ or $A_2$ at least one is mixing.
	\end{enumerate}
\end{theorem}

\begin{theorem}\label{thm:ad'D}
	For any dual $u_1\otimes u_2$-preserving map $\tilde{A}$ with \eqref{ad-assumption}, 
	the following conditions are equivalent.
	\begin{enumerate}
		\item\label{ad'D1}
		$\tilde{A}$ is $(\tilde{\cK},u_1,u_2)$-asymptotically decoupling.
		\item\label{ad'D2}
		There exist a number $i_1\in\{1,2\}$ and a state $x_{0,i_1}\in\cS(\cK_{i_1},u_{i_1})$ 
		such that any state $\tilde{x}\in\cS(\tilde{\cK}, u_1\otimes u_2)$ satisfies 
		\[
		\tilde{A}^n \tilde{x}
		= x_{0,i_1}\otimes\pi_{i_2}\tilde{A}^n \tilde{x} + o(1)
		\quad (n\to\infty),
		\]
		where $i_2$ is an element of $\{1,2\}$ except for $i_1$.
	\end{enumerate}
\end{theorem}

Since Theorem~\ref{thm:adD} follows from Theorem~\ref{thm:ad'D} immediately, 
all we need is to show Theorem~\ref{thm:ad'D}. 
To show Theorem~\ref{thm:ad'D}, 
we introduce a non-asymptotic property on the joint system, namely, \textit{decoupling}: 

\begin{definition}[$(\tilde{\cK}, u_1, u_2)$-decoupling]
	Suppose that a dual $u_1\otimes u_2$-preserving map $\tilde{A}$ satisfies 
	\begin{equation}
		\tilde{A}(\tilde{\cK}\setminus\{ 0 \}) \subset \tilde{\cK}_{\max}\setminus\{ 0 \}.
		\label{d-assumption}
	\end{equation}
	Then $\tilde{A}$ is decoupling 
	if any state $\tilde{x}\in\cS(\tilde{\cK}, u_1\otimes u_2)$ satisfies 
	\begin{equation*}
		\tilde{A}\tilde{x}
		= \pi_1\tilde{A}\tilde{x}\otimes\pi_2\tilde{A}\tilde{x}.
	\end{equation*}
\end{definition}

Decoupling asserts that 
$\tilde{A}$ breaks the correlation between two systems completely. 
To prove Theorem~\ref{thm:ad'D}, 
we give a necessary and sufficient condition of decoupling, 
which can be regarded as the non-asymptotic version of Theorem~\ref{thm:adD}.

\begin{lemma}\label{decoupling}
	For any dual $u_1\otimes u_2$-preserving map $\tilde{A}$ with \eqref{d-assumption}, 
	the following conditions are equivalent.
	\begin{enumerate}
		\item\label{decoupling1}
		$\tilde{A}$ is $(\tilde{\cK},u_1,u_2)$-decoupling.
		\item\label{decoupling2}
		There exist a number $i_1\in\{1,2\}$ and a state $x_{0,i_1}\in\cS(\cK_{i_1},u_{i_1})$ 
		such that any state $\tilde{x}\in\cS(\tilde{\cK}, u_1\otimes u_2)$ satisfies 
		\[
		\tilde{A}\tilde{x}
		= x_{0,i_1}\otimes\pi_{i_2}\tilde{A}\tilde{x},
		\]
		where $i_2$ is an element of $\{1,2\}$ except for $i_1$.
	\end{enumerate}
\end{lemma}
\begin{proof}
	\eqref{decoupling2}$\Rightarrow$\eqref{decoupling1}. 
	This implication follows from the definition.
	\par
	\eqref{decoupling1}$\Rightarrow$\eqref{decoupling2}. 
	Assume condition~\eqref{decoupling1}. 
	Then condition~\eqref{decoupling2} follows from the following steps.
	\par
	\setcounter{count}{0}
	\noindent\textbf{Step~\num.}
	First, we show that any two states 
	$\tilde{x},\tilde{x}' \in \cS(\tilde{\cK}, u_1\otimes u_2)$ satisfy 
	$\pi_1\tilde{A}\tilde{x} = \pi_1\tilde{A}\tilde{x}'$ 
	or $\pi_2\tilde{A}\tilde{x} = \pi_2\tilde{A}\tilde{x}'$ 
	at least one. 
	Let $\tilde{x},\tilde{x}' \in \cS(\tilde{\cK}, u_1\otimes u_2)$. 
	Then the $(a)$ linearity and $(b)$ decoupling of $\tilde{A}$ yield 
	\begin{align*}
		&\frac{1}{2} \pi_1\tilde{A}\tilde{x} \otimes \pi_2\tilde{A}\tilde{x}
		+ \frac{1}{2} \pi_1\tilde{A}\tilde{x}' \otimes \pi_2\tilde{A}\tilde{x}'\\
		\overset{(b)}{=}& \frac{1}{2}\tilde{A}\tilde{x} + \frac{1}{2}\tilde{A}\tilde{x}'
		\overset{(a)}{=} \tilde{A}\Bigl( \frac{1}{2}\tilde{x} + \frac{1}{2}\tilde{x}' \Bigr)\\
		\overset{(b)}{=}& \pi_1\tilde{A}\Bigl( \frac{1}{2}\tilde{x} + \frac{1}{2}\tilde{x}' \Bigr)
		\otimes \pi_2\tilde{A}\Bigl( \frac{1}{2}\tilde{x} + \frac{1}{2}\tilde{x}' \Bigr)\\
		\overset{(a)}{=}& \frac{1}{4} \pi_1\tilde{A}\tilde{x} \otimes \pi_2\tilde{A}\tilde{x}
		+ \frac{1}{4} \pi_1\tilde{A}\tilde{x}' \otimes \pi_2\tilde{A}\tilde{x}'\\
		&+ \frac{1}{4} \pi_1\tilde{A}\tilde{x} \otimes \pi_2\tilde{A}\tilde{x}'
		+ \frac{1}{4} \pi_1\tilde{A}\tilde{x}' \otimes \pi_2\tilde{A}\tilde{x}.
	\end{align*}
	Calculating the above both sides, we obtain 
	\[
	(\pi_1\tilde{A}\tilde{x} - \pi_1\tilde{A}\tilde{x}')
	\otimes (\pi_2\tilde{A}\tilde{x} - \pi_2\tilde{A}\tilde{x}')
	= 0.
	\]
	Therefore, any two states 
	$\tilde{x},\tilde{x}' \in \cS(\tilde{\cK}, u_1\otimes u_2)$ satisfy 
	$\pi_1\tilde{A}\tilde{x} = \pi_1\tilde{A}\tilde{x}'$ 
	or $\pi_2\tilde{A}\tilde{x} = \pi_2\tilde{A}\tilde{x}'$ 
	at least one.
	\par
	\noindent\textbf{Step~\num.}
	Next, we show that there exists a number $i\in\{1,2\}$ 
	such that any state $\tilde{x}$ satisfies 
	$\pi_i\tilde{A}\tilde{x} = \pi_i\tilde{A}\tilde{x}'$ 
	by contradiction. 
	Suppose that there existed two pairs $(\tilde{x}_1,\tilde{x}'_1)$ 
	and $(\tilde{x}_2,\tilde{x}'_2)$ of two states such that 
	\begin{equation}
		\pi_1\tilde{A}\tilde{x}_1 \not= \pi_1\tilde{A}\tilde{x}'_1,\quad
		\pi_2\tilde{A}\tilde{x}_2 \not= \pi_2\tilde{A}\tilde{x}'_2.
		\label{eq01}
	\end{equation}
	Then step~1 implies 
	\begin{equation*}
		\pi_2\tilde{A}\tilde{x}_1 = \pi_2\tilde{A}\tilde{x}'_1,\quad
		\pi_1\tilde{A}\tilde{x}_2 = \pi_1\tilde{A}\tilde{x}'_2.
	\end{equation*}
	Also, thanks to the pigeonhole principle, 
	there exists a number $i_0\in\{1,2\}$ 
	such that two of the following three cases hold: 
	\begin{enumerate}
		\item
		$\pi_{i_0}\tilde{A}(\tilde{x}_1 + \tilde{x}'_1)/2
		= \pi_{i_0}\tilde{A}(\tilde{x}_2 + \tilde{x}'_2)/2$,
		\item
		$\pi_{i_0}\tilde{A}(\tilde{x}_1 + \tilde{x}_2)/2
		= \pi_{i_0}\tilde{A}(\tilde{x}'_1 + \tilde{x}'_2)/2$,
		\item
		$\pi_{i_0}\tilde{A}(\tilde{x}_1 + \tilde{x}'_2)/2
		= \pi_{i_0}\tilde{A}(\tilde{x}'_1 + \tilde{x}_2)/2$.
	\end{enumerate}
	For instance, suppose that cases~(i) and (ii) (which we call \{(i),(ii)\} case) hold. 
	Taking the sum of cases~(i) and (ii), we obtain 
	$\pi_{i_0}\tilde{A}\tilde{x}_1 = \pi_{i_0}\tilde{A}\tilde{x}'_2$. 
	This equation and case~(i) yield 
	$\pi_{i_0}\tilde{A}\tilde{x}'_1 = \pi_{i_0}\tilde{A}\tilde{x}_2$. 
	However, if $i_0=1$, it follows that 
	\[
	\pi_1\tilde{A}\tilde{x}_1
	= \pi_1\tilde{A}\tilde{x}'_2
	= \pi_1\tilde{A}\tilde{x}_2
	= \pi_1\tilde{A}\tilde{x}'_1,
	\]
	which contradicts \eqref{eq01}. 
	Similarly, if $i_0=2$, it follows that 
	\[
	\pi_2\tilde{A}\tilde{x}_2
	= \pi_2\tilde{A}\tilde{x}'_1
	= \pi_2\tilde{A}\tilde{x}_1
	= \pi_2\tilde{A}\tilde{x}'_2,
	\]
	which contradicts \eqref{eq01}. 
	The other cases, namely, \{(i),(iii)\} and \{(ii),(iii)\} cases also have contradictions in the same way. 
	Therefore, there exists a number $i\in\{1,2\}$ 
	such that any two states 
	$\tilde{x},\tilde{x}' \in \cS(\tilde{\cK}, u_1\otimes u_2)$ satisfy 
	$\pi_i\tilde{A}\tilde{x} = \pi_i\tilde{A}\tilde{x}'$.
	\par
	\noindent\textbf{Step~\num.}
	Step~2 implies that 
	$\pi_i\tilde{A}\tilde{x}$ does not depend on $\tilde{x}\in\cS(\tilde{\cK}, u_1\otimes u_2)$. 
	That is, there exists a state $x_{0,i}\in\cS(\cK_i,u_i)$ 
	such that any state $\tilde{x}\in\cS(\tilde{\cK}, u_1\otimes u_2)$ satisfies 
	$\pi_i\tilde{A}\tilde{x} = x_{0,i}$. 
	Thus this equation and condition~\eqref{decoupling1} yield condition~\eqref{decoupling2}.
\end{proof}

Moreover, to prove Theorem~\ref{thm:ad'D}, 
we need the following technical lemmas: 

\begin{lemma}[Compactness]\label{map-compact}
	The set of all $(\cK,u; \cK',u')$-dynamical maps is compact.
\end{lemma}

\begin{lemma}[Invariance of ranges]\label{inv-ran}
	Let $A$ be a dynamical map. 
	If $A_0$ and $A'_0$ are the limits of 
	two convergent subsequences $\{A^n\}_{n\in N\subset\mathbb{N}}$ 
	and $\{A^n\}_{n\in N'\subset\mathbb{N}}$, respectively, 
	then $\Ran A_0=\Ran A'_0$.
\end{lemma}

Lemmas~\ref{map-compact} and \ref{inv-ran} are proved in Appendix~\ref{appA}.

\begin{proof}[Proof of Theorem~$\ref{thm:ad'D}$]
	\eqref{ad'D2}$\Rightarrow$\eqref{ad'D1}. 
	This implication follows from the definition.
	\par
	\eqref{ad'D1}$\Rightarrow$\eqref{ad'D2}. 
	Assume condition~\eqref{ad'D1}. 
	Then condition~\eqref{ad'D2} follows from the following steps.
	\par
	\setcounter{count}{0}
	\noindent\textbf{Step~\num.}
	First, we show that there exist a number $i_1\in\{1,2\}$ 
	and a state $x_{0,i_1}\in\cS(\cK_{i_1},u_{i_1})$ such that 
	the limit $\tilde{A}_0$ of any convergent subsequence of 
	$\{\tilde{A}^n\}_{n\in\mathbb{N}}$ satisfies 
	\[
	\tilde{A}_0(\cdot) = x_{0,i_1}\otimes\pi_{i_2}\tilde{A}_0(\cdot),
	\]
	where $i_2$ is an element of $\{1,2\}$ except for $i_1$. 
	Take an arbitrary convergent subsequence $\{\tilde{A}^n\}_{n\in N\subset\mathbb{N}}$. 
	Then the limit $\tilde{A}_0$ of $\{\tilde{A}^n\}_{n\in N}$ is decoupling. 
	By changing the subscripts if necessary, 
	Lemma~\ref{decoupling} implies that 
	there exists a state $x_{0,1}\in\cS(\cK_1,u_1)$ such that 
	\[
	\tilde{A}_0(\cdot) = x_{0,1}\otimes\pi_2\tilde{A}_0(\cdot).
	\]
	The subscript $1$ and the state $x_{0,1}$ are independent of 
	convergent subsequences of $\{\tilde{A}^n\}_{n\in\mathbb{N}}$, 
	due to Lemma~\ref{inv-ran}. 
	Therefore, the assertion in this step follows.
	\par
	\noindent\textbf{Step~\num.}
	Next, we show that the sequence $\{\pi_1\tilde{A}^n\}_{n\in\mathbb{N}}$ 
	converges to $\braket{u_1\otimes u_2, \cdot}x_{0,1}$. 
	Take an arbitrary convergent subsequence $\{\pi_1\tilde{A}^n\}_{n\in N\subset\mathbb{N}}$. 
	Due to Lemma~\ref{map-compact}, 
	we can take a subsequence $\{\pi_1\tilde{A}^n\}_{n\in N'\subset N}$ 
	such that the sequence $\{\tilde{A}^n\}_{n\in N'}$ converges. 
	Step~1 implies that 
	the limit of $\{\pi_1\tilde{A}^n\}_{n\in N'}$ is $\braket{u_1\otimes u_2, \cdot}x_{0,1}$. 
	Therefore, any convergent subsequence of $\{\pi_1\tilde{A}^n\}_{n\in\mathbb{N}}$ 
	converges to $\braket{u_1\otimes u_2, \cdot}x_{0,1}$. 
	Thus Lemma~\ref{map-compact} implies that the sequence 
	$\{\pi_1\tilde{A}^n\}_{n\in\mathbb{N}}$ converges to $\braket{u_1\otimes u_2, \cdot}x_{0,1}$.
	\par
	\noindent\textbf{Step~\num.}
	Finally, step~2 and condition~\eqref{ad'D1} imply that 
	any state $\tilde{x}\in\cS(\tilde{\cK}, u_1\otimes u_2)$ satisfies 
	\begin{align*}
		\tilde{A}^n \tilde{x}
		&= \pi_1\tilde{A}^n \tilde{x}\otimes\pi_2\tilde{A}^n \tilde{x} + o(1)\\
		&= x_{0,1}\otimes\pi_2\tilde{A}^n \tilde{x} + o(1),
	\end{align*}
	which is just condition~\eqref{ad'D2}.
\end{proof}


So far, we have defined asymptotic decoupling for dynamical maps, 
and have given its necessary and sufficient condition. 
However, to handle stochastic operation including SLOCC, 
we need to address dynamical maps without dual $u_1\otimes u_2$-preserving. 
In the case without dual $u_1\otimes u_2$-preserving, 
the definitions of asymptotic decoupling and mixing are modified as follows.

\begin{definition}[$(\tilde{\cK}, u_1, u_2)$-asymptotic decoupling]\label{def:adP}
	A linear map $\tilde{A}$ with \eqref{ad-assumption} 
	is $(\tilde{\cK}, u_1, u_2)$-asymptotically decoupling 
	if any state $\tilde{x}\in\cS(\tilde{\cK}, u_1\otimes u_2)$ satisfies 
	\begin{equation*}
		\frac{\tilde{A}^n \tilde{x}}{\braket{u_1\otimes u_2, \tilde{A}^n \tilde{x}}}
		= \frac{\pi_1\tilde{A}^n \tilde{x}}{\braket{u_1\otimes u_2, \tilde{A}^n \tilde{x}}}
		\otimes\frac{\pi_2\tilde{A}^n \tilde{x}}{\braket{u_1\otimes u_2, \tilde{A}^n \tilde{x}}}
		+ o(1)
		\quad(n\to\infty).
	\end{equation*}
\end{definition}

\begin{definition}[Mixing]\label{def:mixL}
	A linear map $A$ with $r(A)>0$ is mixing 
	if there exist nonzero $x_0,y_0\in\cV$ such that any $x\in\cV$ satisfies 
	\begin{equation*}
		\lim_{n\to\infty} (r(A)^{-1}A)^n x = \braket{y_0,x}x_0.
	\end{equation*}
	The vectors $x_0$ and $y_0$ are called a stationary vector and a dual stationary vector, respectively.
\end{definition}

Of course, if $\tilde{A}$ is dual $u_1\otimes u_2$-preserving, 
Definitions~\ref{def:adD} and \ref{def:adP} are equivalent to each other. 
If $A$ is a dynamical map, 
Definitions~\ref{def:mixD} and \ref{def:mixL} are also equivalent to each other. 
Indeed, if a $(\cK,u)$-dynamical map $A$ is mixing 
in the sense of Definition~\ref{def:mixL}, 
then Proposition~\ref{d-map} implies $r(A)=1$, 
and we can choose a stationary vector $x_0$ and a dual stationary vector $y_0$ 
as $x_0\in\cS(\cK,u)$ and $y_0=u$. 
Therefore, Definitions~\ref{def:adP} and \ref{def:mixL} 
are extensions of Definitions~\ref{def:adD} and \ref{def:mixD}, respectively. 
By using Definitions~\ref{def:adP} and \ref{def:mixL}, 
Theorem~\ref{thm:adD} is extended as follows: 

\begin{theorem}\label{thm:adP}
	Suppose that the adjoint map of a $\cK_i$-positive map $A_i$ 
	has an eigenvector $y_{0,i}\in\tint(\cK_i)$ for each $i=1,2$. 
	Then the following conditions are equivalent.
	\begin{enumerate}
		\item\label{adP1}
		$A_1\otimes A_2$ is $(\tilde{\cK},u_1,u_2)$-asymptotically decoupling.
		\item\label{adP2}
		$A_1$ or $A_2$ at least one is mixing.
	\end{enumerate}
\end{theorem}

Theorem~\ref{thm:adP} follows from Theorem~\ref{thm:adD} and the next lemma. 
The lemma asserts that $(\tilde{\cK}, u_1, u_2)$-asymptotic decoupling 
does not depend on the unit effect $u_1$ or $u_2$.

\begin{lemma}\label{lem:ad}
	Let $u'_i\in\tint(\cK_i^\ast)$ for each $i=1,2$. 
	If a linear map $\tilde{A}$ with \eqref{ad-assumption} 
	is $(\tilde{\cK}, u_1, u_2)$-asymptotically decoupling, 
	then $\tilde{A}$ is $(\tilde{\cK}, u'_1, u'_2)$-asymptotically decoupling.
\end{lemma}

Lemma~\ref{lem:ad} is proved in Appendix~\ref{appA}.

\begin{proof}[Proof of Theorem~$\ref{thm:adP}$]
	Without loss of generality, we may show the equivalence 
	by assuming that the eigenvector $y_{0,i}\in\tint(\cK_i)$ 
	equals the unit effect $u_i$ for each $i=1,2$, 
	due to Lemma~\ref{lem:ad}. 
	In the case with $y_{0,i}=u_i$, 
	Theorem~\ref{thm:adP} follows from Theorem~\ref{thm:adD}. 
	Therefore, the proof is completed.
\end{proof}

\subsection{Ergodicity}\label{erg}
As proved in Section~\ref{ad}, 
asymptotic decoupling is equivalent to local mixing. 
Hence it is meaningful to investigate criteria of mixing. 
Although spectral criterion is known as a criterion of mixing, 
we give another criterion in the next subsection. 
For this purpose, we investigate ergodicity in this subsection, 
which is a weaker property than mixing. 
In quantum theory, ergodicity is defined for quantum channels. 
In GPTs, however, the definition of dynamical maps on the joint system depends on the cone $\tilde{\cK}$ of the joint system, 
and the cone $\tilde{\cK}$ is not unique. 
For this reason, it is convenient to define ergodicity (and mixing) for linear maps in our context. 
Therefore, we employ the following definition: 

\begin{definition}[Ergodicity]\label{def:ergL}
	A linear map $A$ with $r(A)>0$ is ergodic if 
	there exist nonzero $x_0,y_0\in\cV$ such that any $x\in\cV$ satisfies 
	\begin{equation*}
		\lim_{n\to\infty} \frac{1}{n}\sum_{k=0}^{n-1} (r(A)^{-1}A)^k x = \braket{y_0,x}x_0.
	\end{equation*}
	The vectors $x_0$ and $y_0$ are called a stationary vector and a dual stationary vector, respectively.
\end{definition}

If $A$ is a dynamical map, 
Definition~\ref{def:ergL} can be represented like Definition~\ref{def:q-erg}. 
Indeed, if a $(\cK,u)$-dynamical map $A$ is ergodic, 
then Proposition~\ref{d-map} implies $r(A)=1$, 
and we can choose a stationary vector $x_0$ and a dual stationary vector $y_0$ 
as $x_0\in\cS(\cK,u)$ and $y_0=u$. 
The above relation is similar to 
the relation between Definitions~\ref{def:mixD} and \ref{def:mixL}. 
Now, Definition~\ref{def:ergL} cannot be checked directly by numerical computation. 
Its numerical check is realized by the following technical lemma: 

\begin{proposition}\label{prop:ergL}
	For any linear map $A$ with $r(A)>0$, the following conditions are equivalent.
	\begin{enumerate}
		\item\label{ergL1}
		$A$ is ergodic.
		\item\label{ergL2}
		$\am(A; r(A))=1$ and $\am(A; \lambda)=\gm(A; \lambda)$ 
		for any $\lambda\in\mathbb{C}$ whose absolute value is $r(A)$.
	\end{enumerate}
\end{proposition}

Proposition~\ref{prop:ergL} is proved in Appendix~\ref{appA}. 
If $A$ is $\cK$-positive, 
condition~\eqref{ergL2} in Proposition~\ref{prop:ergL} turns to a simpler condition: 

\begin{proposition}\label{prop:ergP}
	For any $\cK$-positive map $A$ with $r(A)>0$, the following conditions are equivalent.
	\begin{enumerate}
		\item\label{ergP1}
		$A$ is ergodic.
		\item\label{ergP2}
		$\am(A; r(A))=1$.
	\end{enumerate}
\end{proposition}
\begin{proof}
	\eqref{ergP1}$\Rightarrow$\eqref{ergP2}. 
	This implication follows from Proposition~\ref{prop:ergL}.
	\par
	\eqref{ergP2}$\Rightarrow$\eqref{ergP1}. 
	All we need is to show condition~\eqref{ergL2} in Proposition~\ref{prop:ergL}. 
	Suppose that $\lambda\in\mathbb{C}$ satisfies $|\lambda|=r(A)$. 
	If $\lambda$ is not an eigenvalue of $A$, of course, 
	the equation $\am(A; \lambda)=\gm(A; \lambda)=0$ holds. 
	If $\lambda$ is an eigenvalue of $A$, 
	condition~\eqref{ergP2} and Corollary~\ref{coro:pos} imply that 
	$\am(A; \lambda)=\gm(A; \lambda)$. 
	Therefore, condition~\eqref{ergL2} in Proposition~\ref{prop:ergL} holds.
\end{proof}

If $A$ is a dynamical map, 
condition~\eqref{ergP2} in Proposition~\ref{prop:ergP} turns to a weaker condition: 

\begin{proposition}\label{prop:ergD}
	For any dynamical map $A$, the following conditions are equivalent.
	\begin{enumerate}
		\item\label{ergD1}
		$A$ is ergodic.
		\item\label{ergD2}
		$\gm(A; 1)=1$. In other words, $\dim\Ker(A - \id_{\cV})=1$.
	\end{enumerate}
\end{proposition}

Proposition~\ref{prop:ergD} is well-known in quantum theory 
\cite[appendix]{BG}, \cite[Corollary~2]{Burgarth}. 
Here, to prove Proposition~\ref{prop:ergD} 
we use Proposition~\ref{prop:ergP} and \cite[Lemma~5]{BG}. 
Applying \cite[Lemma~5]{BG} to our setting, 
we obtain the following proposition immediately: 

\begin{proposition}\label{prop:fixed}
	If a $(\cK,u)$-dynamical map $A$ 
	has only one fixed point in $\cS(\cK,u)$, 
	then $A$ is ergodic.
\end{proposition}

\begin{proof}[Proof of Proposition~$\ref{prop:ergD}$]
	\eqref{ergD1}$\Rightarrow$\eqref{ergD2}. 
	This implication follows from Proposition~\ref{prop:ergP}.
	\par
	\eqref{ergD2}$\Rightarrow$\eqref{ergD1}. 
	Since $A$ is $(\cK,u)$-dynamical, 
	Propositions~\ref{prop:pos} and \ref{d-map} imply that 
	$A$ has an eigenvector $x_0\in\cK$ associated with $r(A)=1$. 
	Without loss of generality, we may assume $x_0\in\cS(\cK,u)$. 
	Thanks to condition~\eqref{ergD2}, 
	the dynamical map $A$ has the only one fixed point $x_0$ in $\cS(\cK,u)$. 
	Thus Proposition~\ref{prop:fixed} implies condition~\eqref{ergD1}.
\end{proof}

Some researchers should know Propositions~\ref{prop:ergL} and \ref{prop:ergP}, 
but they seem not to be published. 
Proposition~\ref{prop:ergD} should also be known 
because the quantum version of Proposition~\ref{prop:ergD} is well-known, 
and the proof is the same as that of the quantum version. 
However, the above three propositions are not so trivial due to the difference among them. 
If we omitted the proofs of them, 
a large logical gap would occur, 
and thus we have proved the three propositions in this subsection.

\subsection{Mixing}\label{mix}
In this subsection, 
we give a criterion of mixing that is different from spectral criterion. 
Our criterion consists of the following two parts: 
(i) Proposition~\ref{prop:ergD} and 
(ii) the equivalence between mixing for a dynamical map $A$ 
and ergodicity for the two-fold tensor product map $A^{\otimes2}$. 
Since (i) is already known as stated in Section~\ref{erg}, 
(ii) is an essential and additional part. 
First, we begin with spectral criterion: 

\begin{proposition}[Spectral criterion]\label{spectral}
	For any linear map $A$ with $r(A)>0$, 
	the following conditions are equivalent.
	\begin{enumerate}
		\item
		$A$ is mixing.
		\item
		$\am(A; r(A))=1$ and $\am(A; \lambda)=0$ 
		for any $\lambda\not=r(A)$ whose absolute value is $r(A)$.
	\end{enumerate}
\end{proposition}

Spectral criterion is well-known in quantum theory \cite[Theorem~7]{BG}, 
but the above general case seems not to be published. 
However, Proposition~\ref{spectral} can be proved in the same way as the quantum case, 
namely, by considering the Jordan canonical form of $A$. 
Moreover, the statement does not change in the three cases, 
namely, for linear maps, for $\cK$-positive maps, and for dynamical maps. 
This is why we do not prove Proposition~\ref{spectral}.

Next, let us give other conditions equivalent to mixing. 

\begin{theorem}\label{thm:mixL}
	For any linear map $A$ with $r(A)>0$, 
	the following conditions are equivalent.
	\begin{enumerate}
		\item\label{mixL1}
		$A$ is mixing.
		\item\label{mixL2}
		$A^{\otimes2}$ is mixing, and $A$ has the eigenvalue $r(A)$.
		\item\label{mixL3}
		$A^{\otimes2}$ is ergodic, and $A$ has the eigenvalue $r(A)$.
		\item\label{mixL4}
		$\gm(A^{\otimes2}; r(A)^2)=1$, and $A$ has the eigenvalue $r(A)$.
	\end{enumerate}
\end{theorem}

If $A$ is $\cK$-positive, 
Proposition~\ref{prop:pos} guarantees that $A$ has the eigenvalue $r(A)$. 
Thus, in this case, 
conditions in Theorem~\ref{thm:mixL} turn to simpler conditions: 

\begin{theorem}\label{thm:mixP}
	For any $\cK$-positive map $A$ with $r(A)>0$, 
	the following conditions are equivalent.
	\begin{enumerate}
		\item\label{mixP1}
		$A$ is mixing.
		\item\label{mixP2}
		$A^{\otimes2}$ is mixing.
		\item\label{mixP3}
		$A^{\otimes2}$ is ergodic.
		\item\label{mixP4}
		$\gm(A^{\otimes2}; r(A)^2)=1$.
	\end{enumerate}
\end{theorem}

If $A$ is a dynamical map, 
Proposition~\ref{d-map} guarantees $r(A)=1$, 
and thus condition~\eqref{mixP4} can be represented as 
$\dim\Ker(A^{\otimes2} - \id_{\cV^{\otimes2}})=1$. 
Hence Theorem~\ref{thm:mixP} gives a criterion of mixing 
that is a system of linear equations. 
On the other hand, Proposition~\ref{spectral} requires us to solve an eigenequation. 
An eigenequation cannot be solved by finite steps in general, 
but a system of linear equations can always be solved, 
which is an advantage of Theorem~\ref{thm:mixP}. 
Of course, some criteria of mixing that does not require us to solve an eigenequation 
are known for special classes of quantum channels. 
For instance, Burgarth et al.\ \cite[Theorems~10 and 13]{Burgarth} proved that, 
for any quantum channel $\Gamma$ with a strictly positive semi-definite fixed point, 
$\Gamma$ is mixing if and only if $\Gamma^k$ is ergodic for all integers $1\le k\le(\dim\cH)^2$; 
a unital quantum channel $\Gamma$ is mixing if $\Gamma^\ast\circ\Gamma$ is ergodic.

To prove Theorem~\ref{thm:mixL}, 
let us show the following preliminary lemma: 

\begin{lemma}\label{lem:mix}
	If a linear map $A_i$ is mixing for each $i=1,2$, 
	then $A_1\otimes A_2$ is also mixing.
\end{lemma}
\begin{proof}
	Without loss of generality, we may assume $r(A_1)=r(A_2)=1$. 
	Then $r(A_1\otimes A_2)=r(A_1)r(A_2)=1$. 
	For all $i=1,2$ and $x_i\in\cV_i$, 
	\[
	\lim_{n\to\infty} A_i^n x_i = \braket{y_{0,i},x_i}x_{0,i},
	\]
	where $x_{0,i}$ and $y_{0,i}$ are a stationary vector and a dual stationary vector of $A_i$, respectively. 
	Thus any $\tilde{x}\in\cV_1\otimes\cV_2$ satisfies 
	\[
	\lim_{n\to\infty} (A_1\otimes A_2)^n \tilde{x}
	= \braket{y_{0,1}\otimes y_{0,2}, \tilde{x}}x_{0,1}\otimes x_{0,2}.
	\]
	Therefore, $A_1\otimes A_2$ is mixing.
\end{proof}

From now on, the symbol $\vspan(\cX)$ denotes the linear span of a subset $\cX\subset\cV$. 
Also, we recall the definition $\cW(A; \lambda) = \Ker(A - \id_{\cV})$ for a linear map $A$ and its eigenvalue $\lambda$.

\begin{proof}[Proof of Theorem~$\ref{thm:mixL}$]
	\eqref{mixL1}$\Rightarrow$\eqref{mixL2}. 
	This implication follows from Proposition~\ref{spectral} and Lemma~\ref{lem:mix}.
	\par
	\eqref{mixL2}$\Rightarrow$\eqref{mixL3}. 
	This implication follows from the definitions of ergodicity and mixing.
	\par
	\eqref{mixL3}$\Rightarrow$\eqref{mixL4}. 
	This implication follows from Proposition~\ref{prop:ergL}.
	\par
	\eqref{mixL4}$\Rightarrow$\eqref{mixL1}. 
	Assuming condition~\eqref{mixL4}, 
	we verify that $A$ satisfies spectral criterion. 
	Without loss of generality, we may assume $r(A)=1$. 
	First, 	condition~\eqref{mixL4} implies $\gm(A; 1)=1$ 
	because of the inclusion relation $\cW(A; 1)^{\otimes2}\subset\cW(A^{\otimes2}; 1)$. 
	Second, we prove $\am(A; 1)=1$ by contradiction. 
	Due to condition~\eqref{mixL4}, 
	the linear map $A$ has an eigenvector $x_0\in\cV$ associated with $r(A)=1$. 
	Suppose that a nonzero $x\in\cV$ satisfies $Ax=x+x_0$, the equation 
	\begin{align*}
		A^{\otimes2}(x\otimes x_0 - x_0\otimes x)
		&= (x+x_0)\otimes x_0 - x_0\otimes(x+x_0)\\
		&= x\otimes x_0 - x_0\otimes x
	\end{align*}
	holds. Since $x_0^{\otimes2}$ is an eigenvector of $A^{\otimes2}$ associated with $r(A)^2=1$, 
	condition~\eqref{mixL4} implies that 
	$\alpha x_0^{\otimes2} = x\otimes x_0 - x_0\otimes x$ for some $\alpha\in\mathbb{R}$. 
	Thus $\alpha x_0 = x - \alpha'x_0$ for some $\alpha'\in\mathbb{R}$, 
	and the vector $x$ lies in $\vspan(x_0)$. 
	Hence $x$ is an eigenvector of $A$ associated with $r(A)=1$, 
	but the equation $x=Ax=x+x_0$ is a contradiction. 
	Therefore, $\am(A; 1)=1$.
	\par
	Third, by contradiction, we prove that $\am(A; \lambda)=0$ for any $\lambda\not=1$ whose absolute value is one. 
	Suppose that $\lambda\not=1$ is an eigenvalue of $A$ with $|\lambda|=1$, 
	and let $\cV^{\mathbb{C}}$ be the complexification of $\cV$. 
	Then $A$ can be naturally identified with the linear map on $\cV^{\mathbb{C}}$ 
	to be $A(x_1 + \sqrt{-1}x_2) = Ax_1 + \sqrt{-1}Ax_2$ for all $x_1,x_2\in\cV$. 
	Take an eigenvector $z\in\cV^{\mathbb{C}}$ of $A$ associated with $\lambda$. 
	Then $A$ also has the eigenvector $\overline{z}$ associated with $\overline{\lambda}$, 
	where for any $z\in\mathbb{C}$ the symbol $\overline{z}$ denotes the complex conjugate of $z$. 
	The linear map $A^{\otimes2}$ has the eigenvector $z\otimes\overline{z}$ 
	associated with $\lambda\overline{\lambda}=1$, 
	whereas $A^{\otimes2}$ also has the eigenvector $x_0^{\otimes2}$ associated with $r(A)^2=1$. 
	Hence condition~\eqref{mixL4} implies that $x_0^{\otimes2} = \xi z\otimes\overline{z}$ for some $\xi\in\mathbb{C}$. 
	Thus $x_0 = \xi'z$ for some $\xi'\in\mathbb{C}$. 
	Since $x_0$ and $z$ are eigenvectors of $A$ associated with $r(A)=1$ and $\lambda$ respectively, 
	it contradicts the assumption $\lambda\not=1$. 
	Therefore, $\am(A; \lambda)=0$ for any $\lambda\not=1$ whose absolute value is one.
	\par
	From the above argument, it follows that $A$ satisfies spectral criterion. 
	Therefore, Proposition~\ref{spectral} implies that $A$ is mixing.
\end{proof}

Finally, summarizing Theorems~\ref{thm:adD} and \ref{thm:mixP}, 
we have the following theorem immediately: 

\begin{theorem}\label{thm:mix-ad}
	Let $A_1$ and $A_2$ be dynamical maps. 
	If $A_2$ is not mixing, then the following conditions are equivalent.
	\begin{enumerate}
		\item\label{mix-ad1}
		$A_1$ is mixing.
		\item\label{mix-ad2}
		$A_1^{\otimes2}$ is mixing.
		\item\label{mix-ad3}
		$A_1^{\otimes2}$ is ergodic.
		\item\label{mix-ad4}
		$\gm(A_1^{\otimes2}; 1)=1$.
		\item\label{mix-ad5}
		$A_1\otimes A_2$ is $(\tilde{\cK}, u_1, u_2)$-asymptotically decoupling.
	\end{enumerate}
\end{theorem}

\subsection{Continuous-time evolution}\label{continuous}
On the basis of results in Sections~\ref{ad} and \ref{mix}, 
we prove the continuous version of Theorems~\ref{thm:adD}, \ref{thm:ad'D}, and \ref{thm:mixP}. 
For simplicity, we use dual $u$-preserving in this subsection, 
but one can verify the continuous versions without dual $u$-preserving. 
For an initial state $x\in\cS(\cK,u)$, 
the state at time $t>0$ is denoted by $A^{(t)}x$ 
in the same way as Section~\ref{q-continuous}. 
Moreover, the family $\{A^{(t)}\}_{t>0}$ of linear maps 
is a $C_0$ semigroup: 
\begin{itemize}
	\item
	$A^{(t)}A^{(s)} = A^{(t+s)}$ for all $t,s>0$,
	\item
	$A^{(t)} \to \id_{\cV}$ as $t\downarrow0$.
\end{itemize}
It can be easily checked that 
any $C_0$ semigroup $\{A^{(t)}\}_{t>0}$ is right-continuous everywhere.

In order to define asymptotic decoupling 
for $C_0$ semigroups of (not necessarily $\tilde{\cK}$-positive) linear maps, 
we introduce the set 
\[
\cD_t(\{\tilde{A}^{(t)}\}_{t>0})
\coloneqq \Set{\tilde{x}\in\tilde{\cK}\setminus\{ 0 \} | 
\tilde{A}^{(s)} \tilde{x} \in \tilde{\cK}_{\max}\setminus\{ 0 \}\ (\forall s\ge t)}
\]
for a $C_0$ semigroup $\{\tilde{A}^{(t)}\}_{t>0}$ of linear maps. 
The set $\cD_t(\{\tilde{A}^{(t)}\}_{t>0})$ is the continuous counterpart of $\cD_n(\tilde{A})$. 
The continuous version of asymptotic decoupling is defined 
by replacing $\cD_n(\tilde{A})$, $\tilde{A}^n$, $n\to\infty$, and \eqref{ad-assumption} 
with $\cD_t(\{\tilde{A}^{(t)}\}_{t>0})$, $\tilde{A}^{(t)}$, $t\to\infty$, and 
\begin{equation}
	\bigcup_{t>0} \cD_t(\{\tilde{A}^{(t)}\}_{t>0}) = \tilde{\cK}\setminus\{ 0 \},
	\label{c-ad-assumption}
\end{equation}
respectively. 
The continuous version of mixing is also defined 
by replacing $A^n$ and $n\to\infty$ with $A^{(t)}$ and $t\to\infty$, respectively. 
In this setting, the continuous versions of Theorems~\ref{thm:adD} and \ref{thm:ad'D} are below: 

\begin{theorem}\label{thm:c-adD}
	For any two $C_0$ semigroups $\{A_1^{(t)}\}_{t>0}$ and $\{A_2^{(t)}\}_{t>0}$ of dynamical maps, 
	the following conditions are equivalent.
	\begin{enumerate}
		\item\label{c-adD1}
		$\{A_1^{(t)}\otimes A_2^{(t)}\}_{t>0}$ is asymptotically decoupling.
		\item\label{c-adD2}
		$\{A_1^{(t)}\}_{t>0}$ or $\{A_2^{(t)}\}_{t>0}$ at least one is mixing.
	\end{enumerate}
\end{theorem}

\begin{theorem}\label{thm:c-ad'D}
	For any $C_0$ semigroup $\{\tilde{A}^{(t)}\}_{t>0}$ of dual $u_1\otimes u_2$-preserving maps with \eqref{c-ad-assumption}, 
	the following conditions are equivalent.
	\begin{enumerate}
		\item\label{c-ad'D1}
		$\{\tilde{A}^{(t)}\}_{t>0}$ is asymptotically decoupling.
		\item\label{c-ad'D2}
		There exist a number $i_1\in\{1,2\}$ and a state $x_{0,i_1}\in\cS(\cK_{i_1}, u_{i_1})$ 
		such that any state $\tilde{x}\in\cS(\tilde{\cK}, u_1\otimes u_2)$ satisfies 
		\[
		\tilde{A}^{(t)}\tilde{x}
		= x_{0,i_1}\otimes\pi_{i_2}\tilde{A}^{(t)}\tilde{x} + o(1)
		\quad (t\to\infty),
		\]
		where $i_2$ is an element of $\{1,2\}$ except for $i_1$.
	\end{enumerate}
\end{theorem}

Since Theorem~\ref{thm:c-adD} follows from Theorem~\ref{thm:c-ad'D} immediately, 
we show Theorem~\ref{thm:c-ad'D} below by using Theorem~\ref{thm:ad'D}.

\begin{proof}[Proof of Theorem~$\ref{thm:c-ad'D}$]
	\eqref{c-ad'D2}$\Rightarrow$\eqref{c-ad'D1}. 
	This implication follows from the definition.
	\par
	\eqref{c-ad'D1}$\Rightarrow$\eqref{c-ad'D2}. 
	Assume condition~\eqref{c-ad'D1}. 
	Since $(\tilde{A}^{(1)})^n = \tilde{A}^{(n)}$ for all $n\in\mathbb{N}$, 
	the linear map $\tilde{A}^{(1)}$ is $(\tilde{\cK}, u_1, u_2)$-asymptotically decoupling. 
	Thus Theorem~\ref{thm:ad'D} implies that 
	$\tilde{A}^{(1)}$ satisfies condition~\eqref{ad'D2} in Theorem~\ref{thm:ad'D}. 
	Without loss of generality, we may assume $(i_1,i_2)=(1,2)$. 
	Thus any state $\tilde{x}\in\cS(\tilde{\cK}, u_1\otimes u_2)$ satisfies 
	\begin{gather}
		\tilde{A}^{(n)}\tilde{x}
		= x_{0,1}\otimes\pi_2\tilde{A}^{(n)}\tilde{x} + o(1), \nonumber\\
		\pi_1\tilde{A}^{(n)}\tilde{x} = x_{0,1} + o(1). \label{eq02}
	\end{gather}
	Now, let $A_{0,1}$ be the dynamical map defined as 
	$A_{0,1}x = \braket{u_1,x}x_{0,1}$ for all $x\in\cV_1$. 
	Take an arbitrary state $\tilde{x}\in\cS(\tilde{\cK}, u_1\otimes u_2)$ 
	and a sufficiently large $n_0\in\mathbb{N}$ satisfying 
	$\tilde{A}^{(t)}\tilde{x}\in\cS(\tilde{\cK}_{\max}, u_1\otimes u_2)$ for all $t\ge n_0$. 
	Then, for all $t\ge n_0$, 
	the real numbers $n = n(t) \coloneqq \lfloor{t}\rfloor - n_0$ 
	and $\delta = \delta(t) \coloneqq t - n(t)$ satisfy 
	\begin{align}
		\pi_1\tilde{A}^{(t)}\tilde{x}
		&= \pi_1\tilde{A}^{(n)}(\tilde{A}^{(\delta)}\tilde{x} - \tilde{x})
		+ \pi_1\tilde{A}^{(n)}\tilde{x}\nonumber\\
		&= \pi_1\tilde{A}^{(n)}(\tilde{A}^{(\delta)}\tilde{x} - \tilde{x})
		+ x_{0,1} + o(1), \label{eq03}
	\end{align}
	where the last equality follows from \eqref{eq02}. 
	Moreover, we have 
	\begin{align}
		\| \pi_1\tilde{A}^{(n)}(\tilde{A}^{(\delta)}\tilde{x} - \tilde{x}) \|
		&= \| (\pi_1\tilde{A}^{(n)} - A_{0,1})(\tilde{A}^{(\delta)}\tilde{x} - \tilde{x}) \|\nonumber\\
		&\le \| \pi_1\tilde{A}^{(n)} - A_{0,1} \|
		\cdot \| \tilde{A}^{(\delta)}\tilde{x} - \tilde{x} \|, \label{eq04}
	\end{align}
	where the norm of the first factor of \eqref{eq04} 
	denotes the operator norm based on 
	the two norms on $\cV_1\otimes\cV_2$ and $\cV_1$: 
	for a linear map $B: \cV_1\otimes\cV_2 \to \cV_1$, 
	$\|B\| \coloneqq \sup_{\|\tilde{x}\|\le1} \|B\tilde{x}\|$. 
	The first factor of \eqref{eq04} vanishes as $t\to\infty$. 
	The second factor of \eqref{eq04} is bounded, 
	since the inequality $\delta\ge n_0$ guarantees 
	$\tilde{A}^{(\delta)}\tilde{x}\in\cS(\tilde{\cK}_{\max}, u_1\otimes u_2)$ 
	and the set $\cS(\tilde{\cK}, u_1\otimes u_2)$ is compact. 
	Thus \eqref{eq04} vanishes as $t\to\infty$, 
	whence \eqref{eq03} turns to 
	$\pi_1\tilde{A}^{(t)}\tilde{x} = x_{0,1} + o(1)$. 
	This equation and condition~\eqref{q-c-ad'D1} imply condition~\eqref{q-c-ad'D2}.
\end{proof}

Next, let us consider the continuous version of Theorem~\ref{thm:mixP} 
for $C_0$ semigroups of dynamical maps. 
For this purpose, we need to define ergodicity 
for $C_0$ semigroups of (not necessarily $\cK$-positive) linear maps: 

\begin{definition}[Ergodicity]
	A $C_0$ semigroup $\{A^{(t)}\}_{t>0}$ of dual $u$-preserving maps is ergodic 
	if there exists a nonzero $x_0\in\cV$ such that 
	any state $x\in\cS(\cK,u)$ satisfies 
	\begin{equation*}
		\lim_{t\to\infty} \frac{1}{t}\int_0^t A^{(s)}x\,ds = x_0.
	\end{equation*}
	The vector $x_0$ is called the stationary vector.
\end{definition}

If all $A^{(t)}$ are $(\cK,u)$-dynamical, 
the stationary vector $x_0$ is a state in $\cS(\cK,u)$. 
Now, the continuous version of Theorem~\ref{thm:mixP} is below: 

\begin{theorem}\label{thm:c-mixD}
	For any $C_0$ semigroup $\{A^{(t)}\}_{t>0}$ of dynamical maps, 
	the following conditions are equivalent.
	\begin{enumerate}
		\item\label{c-mixD1}
		$\{A^{(t)}\}_{t>0}$ is mixing.
		\item\label{c-mixD2}
		$\{(A^{(t)})^{\otimes2}\}_{t>0}$ is mixing.
		\item\label{c-mixD3}
		$\{(A^{(t)})^{\otimes2}\}_{t>0}$ is ergodic.
		\item\label{c-mixD4}
		$A^{(\epsilon)}$ is mixing for some $\epsilon>0$.
		\item\label{c-mixD5}
		$(A^{(\epsilon)})^{\otimes2}$ is mixing for some $\epsilon>0$.
		\item\label{c-mixD6}
		$(A^{(\epsilon)})^{\otimes2}$ is ergodic for some $\epsilon>0$.
	\end{enumerate}
\end{theorem}

We show Theorem~\ref{thm:c-mixD} by using Theorem~\ref{thm:mixP}. 
To prove Theorem~\ref{thm:c-mixD}, 
we need the following preliminary lemma: 

\begin{lemma}\label{lem:c-erg}
	For any ergodic $C_0$ semigroup $\{A^{(t)}\}_{t>0}$ of dual $u$-preserving maps, 
	there exists $\epsilon>0$ such that $A^{(\epsilon)}$ is ergodic.
\end{lemma}
\begin{proof}
	Since $\lim_{\epsilon\downarrow0} \epsilon^{-1}\int_0^\epsilon A^{(s)}\,ds = \id_{\cV}$, 
	there exists $\epsilon>0$ such that $\epsilon^{-1}\int_0^\epsilon A^{(s)}\,ds$ is invertible. 
	Take such a positive number $\epsilon$, 
	and put $B = \epsilon^{-1}\int_0^\epsilon A^{(s)}\,ds$. 
	Since $\{A^{(t)}\}_{t>0}$ is ergodic, 
	there exists a nonzero $x_0\in\cV$ such that any state $x\in\cS(\cK,u)$ satisfies 
	\[
	x_0 = \lim_{n\to\infty} \frac{1}{n\epsilon}\int_0^{n\epsilon} A^{(s)}x\,ds
	= \lim_{n\to\infty} \frac{1}{n\epsilon}\sum_{k=0}^{n-1} \int_0^\epsilon A^{(s+k\epsilon)}x\,ds
	= \lim_{n\to\infty} \frac{1}{n}B\sum_{k=0}^{n-1} A^{(k\epsilon)}x.
	\]
	Thus any state $x\in\cS(\cK,u)$ satisfies 
	\[
	B^{-1}x_0 = \lim_{n\to\infty} \frac{1}{n}\sum_{k=0}^{n-1} A^{(k\epsilon)}x,
	\]
	whence $A^{(\epsilon)}$ is ergodic.
\end{proof}

\begin{proof}[Proof of Theorem~$\ref{thm:c-mixD}$]
	\eqref{c-mixD1}$\Rightarrow$\eqref{c-mixD2}$\Rightarrow$\eqref{c-mixD3}. 
	These implications can be proved 
	in the same way as the proof of Theorem~\ref{thm:mixL}.
	\par
	\eqref{c-mixD3}$\Rightarrow$\eqref{c-mixD6}$\Rightarrow$\eqref{c-mixD5}$\Rightarrow$\eqref{c-mixD4}. 
	These implications follow from Lemma~\ref{lem:c-erg} and Theorem~\ref{thm:mixP}.
	\par
	\eqref{c-mixD4}$\Rightarrow$\eqref{c-mixD1}. 
	Assume condition~\eqref{c-mixD4}: 
	there exists a state $x_0\in\cS(\cK,u)$ such that any state $x\in\cS(\cK,u)$ satisfies 
	\[
	\lim_{n\to\infty} A^{(n\epsilon)}x = x_0.
	\]
	Let $A_0$ be the $(\cK,u)$-dynamical map defined as 
	$A_0 x = \braket{u,x}x_0$ for all $x\in\cV$. 
	Also, let $n = n(t) := \lfloor{t/\epsilon}\rfloor$ and $\delta = \delta(t) := t - n(t)\epsilon$ for all $t\ge\epsilon$. 
	Then any state $x\in\cS(\cK,u)$ satisfies 
	\begin{align*}
		A^{(t)}x
		&= A^{(n\epsilon)}(A^{(\delta)}x - x) + A^{(n\epsilon)}x\\
		&= (A^{(n\epsilon)} - A_0)(A^{(\delta)}x - x) + x_0 + o(1).
	\end{align*}
	Since $(A^{(n\epsilon)} - A_0)(A^{(\delta)}x - x) = o(1)$ 
	can be proved in the same way as the proof of $\text{\eqref{eq04}}=o(1)$, 
	we obtain $A^{(t)}x \to x_0$ as $t\to\infty$, 
	which is just condition~\eqref{c-mixD1}.
\end{proof}

\section{Application to Perron-Frobenius theory}\label{PF}
In this section, we apply Theorem~\ref{thm:mixL} to Perron-Frobenius theory 
involving the Perron-Frobenius theorem. 
The Perron-Frobenius theorem is a famous theorem in linear algebra 
and has many applications in applied mathematics. 
For example, assuming \textit{irreducibility} which is a property studied in Perron-Frobenius theory, 
the references \cite{W-H,HY} gave 
large deviation, moderate deviation, and the central limit theorem 
in analyzing hidden Markovian processes. 
Moreover, assuming \textit{primitivity} which is also a property studied in Perron-Frobenius theory, 
they gave a calculus formula of the asymptotic variance in the central limit theorem.
These facts motivate us to study irreducibility and primitivity.

Another motivation to study irreducibility and primitivity is that 
irreducibility in classical probability theory can be easily checked. 
To see this fact, 
we state irreducibility for a classical channel (stochastic matrix) $W$: 
\textit{for all integers $1\le i,j\le d$, 
there exists $n\in\mathbb{N}$ such that $\braket{j |W^n| i}>0$, 
where $\{ \ket{i} \}_{i=1}^d$ is the standard basis of $\mathbb{R}^d$.} 
This condition is called \textit{irreducibility} \cite{W-H}. 
This classical irreducibility can be easily checked 
by investigating the supports of outputs for the finite number of inputs. 
Hence one often uses irreducibility rather than ergodicity, 
although irreducibility is close to ergodicity. 
The closeness is clarified in Definition~\ref{def:irr} and Proposition~\ref{prop:irrP}.

However, if we generalize the classical definition 
to the case with a general proper cone $\cK$, 
it is difficult to check irreducibility in general. 
To explain this difficulty, let us introduce a few terms. 
We say that $x\in\cK$ is \textit{extreme} 
if for all $x_1,x_2\in\cK$ 
the relation $x=x_1+x_2$ implies $x_1,x_2\in\vspan(x)$ \cite[section 2]{Vander}. 
Also, an \textit{extreme ray} of a cone $\cK$ 
is defined as the subset $\set{\alpha x | \alpha\ge0}\subset\cK$ 
with a nonzero extreme vector $x\in\cK$. 
For a $\cK$-positive map $A$, 
it is natural to define irreducibility as follows: 
\textit{for all nonzero extreme $x\in\cK$ and $y\in\cK^\ast$, 
there exists $n\in\mathbb{N}$ such that $\braket{y,A^n x}>0$.} 
If the number of extreme rays of $\cK$ is finite, irreducibility can be easily checked. 
The cone $\cK=[0,\infty)^d$ in classical probability theory has certainly exact $d$ extreme rays. 
However, the number of extreme rays of $\cK$ is not finite in general, 
and hence it is difficult to check irreducibility in general.

In order to apply Theorem~\ref{thm:mixL} to Perron-Frobenius theory, 
we introduce $\cK$-irreducibility and $\cK$-primitivity 
as stronger conditions than ergodicity and mixing, respectively. 
The definition of $\cK$-irreducibility is equivalent to that in the previous paragraph 
(see Proposition~\ref{prop:irrP}). 
From now on, let $\cK$ and $\tilde{\cK}$ be proper cones of $\cV$ and $\cV^{\otimes2}$ 
satisfying $\tilde{\cK}_{\min} \subset \tilde{\cK} \subset \tilde{\cK}_{\max}$. 
Then $\cK$-irreducibility and $\cK$-primitivity are defined as follows: 

\begin{definition}[$\cK$-irreducibility]\label{def:irr}
	A linear map $A$ is $\cK$-irreducible 
	if the following conditions hold.
	\begin{itemize}
		\item
		$A$ is ergodic.
		\item
		The interiors of $\cK$ and $\cK^\ast$ contain 
		a stationary vector $x_0$ and a dual stationary vector $y_0$, respectively.
	\end{itemize}
\end{definition}

\begin{definition}[$\cK$-primitivity]\label{def:prim}
	A linear map $A$ is $\cK$-primitive 
	if the following conditions hold.
	\begin{itemize}
		\item
		$A$ is mixing.
		\item
		The interiors of $\cK$ and $\cK^\ast$ contain 
		a stationary vector $x_0$ and a dual stationary vector $y_0$, respectively.
	\end{itemize}
\end{definition}

In the above definitions, $\cK$-irreducibility and $\cK$-primitivity 
are defined for linear maps, 
but it is usual to define $\cK$-irreducibility and $\cK$-primitivity 
for $\cK$-positive maps. 
If $A$ is $\cK$-positive, 
Definitions~\ref{def:irr} and \ref{def:prim} are the same as usual ones. 
Since $\cK$-irreducibility and $\cK$-primitivity 
are clearly related to ergodicity and mixing, 
Theorem~\ref{thm:mixL} and Lemma~\ref{cone-interior} in Appendix~\ref{appC} 
yield the following corollary immediately: 

\begin{corollary}\label{coro:I-P}
	Let $A$ be a linear map having the eigenvalue $r(A)$. 
	Then $A$ is $\cK$-primitive if and only if 
	$A^{\otimes2}$ is $\tilde{\cK}$-irreducible.
\end{corollary}

\begin{table}[t]
	\centering
	\renewcommand{\arraystretch}{1.2}
	\caption{Relations among ergodicity, mixing, 
	$\cK$-irreducibility, and $\cK$-primitivity. 
	The symbols $A$, $x_0$, and $x$ are a $(\cK,u)$-dynamical map, 
	the stationary state, an initial state, respectively. 
	The limit $n\to\infty$ is taken in the upper table.}
	\label{T02}
	\vskip2ex
	Classification according to definitions
	\par
	\begin{tabular}{|c||c|c|}
		\hline
		   &$x_0\in\cK$&$x_0\in\tint(\cK)$\\
		\hline\hline
		$(1/n)\sum_{k=0}^{n-1} A^k x\to x_0$&Ergodic&Irreducible\\
		\hline
		$A^n x\to x_0$&Mixing&Primitive\\
		\hline
	\end{tabular}
	\vskip2ex
	Classification according to Proposition~\ref{prop:ergD} and Theorem~\ref{thm:mixP}
	\par
	\begin{tabular}{|c||c|c|}
		\hline
		   &$x_0\in\cK$&$x_0\in\tint(\cK)$\\
		\hline\hline
		$\dim\Ker(A - \id_{\cV})=1$&Ergodic&Irreducible\\
		\hline
		$\dim\Ker(A^{\otimes2} - \id_{\cV^{\otimes2}})=1$&Mixing&Primitive\\
		\hline
	\end{tabular}
\end{table}

The upper table in Table~\ref{T02} summarizes the relation 
among ergodicity, mixing, irreducibility, and primitivity 
according to their definitions. 
The lower table in Table~\ref{T02} summarizes the relation 
among them according to Proposition~\ref{prop:ergD} and Theorem~\ref{thm:mixP}. 
Now, we use Corollary~\ref{coro:I-P} to obtain conditions equivalent to $\cK$-primitivity. 
In preceding studies, there are only a few conditions equivalent to $\cK$-primitivity. 
On the other hand, many conditions equivalent to $\cK$-irreducibility are known. 
For instance, the following conditions are known: 

\begin{proposition}\label{prop:irrP}
	For any $\cK$-positive map $A$, 
	the following conditions are equivalent.
	\renewcommand{\theenumi}{I\arabic{enumi}}
	\begin{enumerate}
		\item\label{mean}
		$A$ is $\cK$-irreducible.
		\item\label{exp}
		For any $x\in\cK\setminus\{ 0 \}$, 
		there exists $t>0$ such that $e^{tA}x\in\tint(\cK)$.
		\item\label{improve}
		Any $x\in\cK\setminus\{ 0 \}$ satisfies $(\id_{\cV} + A)^{d-1}x\in\tint(\cK)$.
		\item\label{ineq}
		If $x\in\cK\setminus\{ 0 \}$ and $\alpha\ge0$ satisfy $\alpha x - Ax\in\cK$, 
		then $x\in\tint(\cK)$.
		\item\label{evec1}
		$A$ has no eigenvector on the boundary of $\cK$.
		\item\label{evec2}
		$A$ and $A^\ast$ have eigenvectors $x_0\in\tint(\cK)$ and $y_0\in\tint(\cK^\ast)$ 
		associated with $r(A)>0$, respectively, and $\gm(A; r(A))=1$.
		\item[\eqref{exp}$'$]
		For any nonzero extreme $x\in\cK$ and $y\in\cK^\ast$, 
		there exists $n\in\mathbb{N}$ such that $\braket{y,A^n x}>0$.
	\end{enumerate}
	Here $d$ denotes the dimension of $\cV$.
\end{proposition}

Although most equivalences of conditions in Proposition~\ref{prop:irrP} can be found 
in \cite{Vander,Barker,Barker-Schneider}, for completeness, 
we prove Proposition~\ref{prop:irrP} in Appendix~\ref{appB}. 
Combining Corollary~\ref{coro:I-P} and Proposition~\ref{prop:irrP}, 
we obtain the following conditions equivalent to $\cK$-primitivity: 

\begin{theorem}\label{thm:prim}
	For any $\cK$-positive map $A$ 
	such that $A^{\otimes2}$ is $\tilde{\cK}$-positive, 
	the following conditions are equivalent.
	\renewcommand{\theenumi}{P\arabic{enumi}}
	\begin{enumerate}
		\item
		$A$ is $\cK$-primitive.
		\item\label{p-exp}
		For any $\tilde{x}\in\tilde{\cK}\setminus\{ 0 \}$, 
		there exists $t>0$ such that $e^{tA^{\otimes2}}\tilde{x} \in \tint(\tilde{\cK})$.
		\item
		Any $\tilde{x}\in\tilde{\cK}\setminus\{ 0 \}$ satisfies 
		$(\id_{\cV^{\otimes2}} + A^{\otimes2})^{d^2-1}\tilde{x} \in \tint(\tilde{\cK})$.
		\item
		If $\tilde{x}\in\tilde{\cK}\setminus\{ 0 \}$ and $\alpha\ge0$ satisfy 
		$\alpha\tilde{x} - A^{\otimes2}\tilde{x} \in \tilde{\cK}$, 
		then $\tilde{x}\in\tint(\tilde{\cK})$.
		\item
		$A^{\otimes2}$ has no eigenvectors on the boundary of $\tilde{\cK}$.
		\item
		$A^{\otimes2}$ and $(A^{\otimes2})^\ast$ have eigenvectors 
		in the interiors of $\tilde{\cK}$ and $\tilde{\cK}^\ast$ 
		associated with $r(A)^2>0$, respectively, and $\gm(A^{\otimes2}; r(A)^2)=1$.
		\item[\eqref{p-exp}$'$]
		For any nonzero extreme $\tilde{x}\in\tilde{\cK}$ and $\tilde{y}\in\tilde{\cK}^\ast$, 
		there exists $n\in\mathbb{N}$ such that $\braket{\tilde{y}, (A^{\otimes2})^n \tilde{x}}>0$.
	\end{enumerate}
	Here $d$ denotes the dimension of $\cV$.
\end{theorem}

\begin{figure}[t]
	\centering
	$\bordermatrix{
		 &0  &1&2&3\cr
		0&0  &1&0&1\cr
		1&1/2&0&0&0\cr
		2&1/2&0&0&0\cr
		3&0  &0&1&0
	}$
	\hspace{1em}
	\includegraphics[scale=0.9]{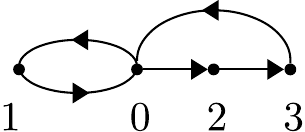}
	\hspace{1em}
	$\bordermatrix{
		 &0&1&2&3\cr
		0&0&1&0&1\cr
		1&1&0&0&0\cr
		2&1&0&0&0\cr
		3&0&0&1&0
	}$
	\caption{The left matrix is a stochastic matrix, 
	the center graph is the directed graph corresponding to the stochastic matrix, and 
	the right matrix is the adjacency matrix of the center graph.}
	\label{F1}
\end{figure}

\begin{remark}
If $A$ is a classical channel (stochastic matrix), 
Corollary~\ref{coro:I-P} can also be derived from an existing result in graph theory. 
In order to explain it, 
we rewrite irreducibility and primitivity for stochastic matrices 
to terms in graph theory. 
A directed graph is called \textit{strongly connected} 
if all two vertices are connected by a path in each direction \cite{McAndrew}. 
A stochastic matrix $W$ is irreducible if and only if 
the directed graph corresponding to $W$ is strongly connected \cite{Meyer}. 
The directed graph corresponding to $W$ has the $d$ vertices $0,\ldots,d-1$, 
and it has directed edges $(j,i)$ if $\braket{j|W|i}>0$. 
Figure~\ref{F1} is an instance of a stochastic matrix, 
the directed graph corresponding to the stochastic matrix, 
and the adjacency matrix of the graph. 
Since the path $0\mapsto1\mapsto0\mapsto2\mapsto3\mapsto0$ passes 
all the vertices and returns to the first vertex, 
the graph in Figure~\ref{F1} is strongly connected. 
Furthermore, the \textit{period} is defined as 
the greatest common divisor of the length of all cycles contained in a strongly connected graph. 
The period of an irreducible stochastic matrix is defined as 
the period of the directed graph corresponding to the stochastic matrix \cite{Kitchens}. 
The reference \cite{Meyer} call it the index of imprimitivity. 
The period of an irreducible stochastic matrix $W$ equals $1$ 
if and only if $W$ is primitive 
(for the only if part, see \cite{Meyer,Kitchens}; the if part can be proved from the definition). 
In graph theory, the \textit{tensor product} of two directed graph $G_1$ and $G_2$ is defined as 
the directed graph corresponding to the adjacency matrix $A_1\otimes A_2$, 
where $A_1$ and $A_2$ are the adjacency matrices of $G_1$ and $G_2$, respectively. 
McAndrew \cite{McAndrew} called it the \textit{product} simply 
and proved the following statement \cite[Theorem~1]{McAndrew}: 
\textit{if two directed graphs $G_1$ and $G_2$ are strongly connected, 
then their tensor product has exact $d_{12}$ strongly connected components, 
where $d_{12}$ is the greatest common divisor of the periods of $G_1$ and $G_2$.} 
If both $G_1$ and $G_2$ are the directed graph corresponding to an irreducible stochastic matrix $W$, 
the above statement implies Corollary~\ref{coro:I-P} 
because $d_{12}$ equals the period of the directed graph corresponding to $W$ in this case.
\end{remark}

Finally, as a simple application of Corollary~\ref{coro:I-P}, 
we show that $\cK$-irreducibility and $\cK$-primitivity 
are preserved under a suitable deformation of a $\cK$-positive map. 
The deformation is used in considering a cumulant generating function \cite{HY}.

\begin{corollary}
	Let $\Omega$ be a non-empty finite set, 
	$A_\omega$ be a $\cK$-positive map, 
	and $a_\omega$ for each $\omega\in\Omega$. 
	If $A \coloneqq \sum_{\omega\in\Omega} A_\omega$ is $\cK$-irreducible, 
	then so is $A_a \coloneqq \sum_\omega a_\omega A_\omega$. 
	If $A \coloneqq \sum_{\omega\in\Omega} A_\omega$ is $\cK$-primitive, 
	then so is $A_a \coloneqq \sum_{\omega\in\Omega} a_\omega A_\omega$.
\end{corollary}
\begin{proof}
	First, assume that the $\cK$-positive map $A$ is $\cK$-irreducible, 
	and let us use condition~\eqref{ineq}. 
	Let $x\in\cK\setminus\{ 0 \}$, $\alpha\ge0$, 
	$\alpha x - A_a x\in\cK$, 
	and $a_{\min} \coloneqq \min_{\omega\in\Omega} a_\omega > 0$. 
	Then $A_a x - a_{\min} Ax\in\cK$ and thus 
	$\alpha x - a_{\min} Ax = (\alpha x - A_a x) + (A_a x - a_{\min} Ax) \in \cK$. 
	Since $A$ is $\cK$-irreducible, we obtain $x\in\tint(\cK)$. 
	Therefore, $A_a$ is also $\cK$-irreducible.
	\par
	Next, assume that $A$ is $\cK$-primitive. 
	Then Corollary~\ref{coro:I-P} implies that the $\tilde{\cK}_{\min}$-positive map 
	$A^{\otimes2} = \sum_{\omega_1,\omega_2\in\Omega} A_{\omega_1}\otimes A_{\omega_2}$ 
	is $\tilde{\cK}_{\min}$-irreducible. 
	The first assertion in this theorem implies that 
	$A_a^{\otimes2} = \sum_{\omega_1,\omega_2\in\Omega} a_{\omega_1}a_{\omega_2} A_{\omega_1}\otimes A_{\omega_2}$ 
	is also $\tilde{\cK}_{\min}$-irreducible. 
	Thus it follows that $A_a$ is $\cK$-primitive from Corollary~\ref{coro:I-P}.
\end{proof}

\section{Conclusion}\label{conclusion}
We have addressed asymptotic properties for Markovian dynamics, 
and proved that asymptotic decoupling is equivalent to local mixing. 
Since this equivalence shows the importance to investigate criteria of mixing, 
we have given a criterion of mixing 
that is a system of linear equations 
by showing that mixing for dynamics is equivalent to 
ergodicity for the two-fold tensor product of dynamics. 
As a simple application, 
we have stated conditions equivalent to primitivity 
by using conditions equivalent to irreducibility. 
Irreducibility and primitivity guarantee the existence of Perron-Frobenius eigenvalues, 
which characterize the asymptotic performance of observed values in analyzing a hidden Markovian process. 
Although the structure of linear maps with the invariance of a cone is studied in Perron-Frobenius theory, 
our result would be useful in Perron-Frobenius theory because of the generality of our setting.

\ack
The authors are grateful to Ziyu Liu for discussing the continuous case together. 
YY was supported by Japan Society for the Promotion of Science (JSPS) 
Grant-in-Aid for JSPS Fellows No.\ 19J20161. 
MH was supported in part by a JSPS Grant-in-Aids for 
Scientific Research (A) No.\ 17H01280 and for Scientific Research (B) No.\ 16KT0017, 
and Kayamori Foundation of Information Science Advancement.

\appendix
\section*{Appendices}
\section{Proofs of technical lemmas}\label{appA}
In this appendix, we prove Lemmas~\ref{map-compact}, \ref{inv-ran}, and \ref{lem:ad} 
and Proposition~\ref{prop:ergL}. 
Also, we show that any two dynamical maps $A_1$ and $A_2$ satisfy 
$\cD_1(A_1\otimes A_2) = \tilde{\cK}\setminus\{ 0 \}$, 
which implies the assumption \eqref{ad-assumption}.

\begin{lemma}\label{lem:A1}
	Any two dynamical maps $A_1$ and $A_2$ satisfy 
	$\cD_1(A_1\otimes A_2) = \tilde{\cK}\setminus\{ 0 \}$.
\end{lemma}
\begin{proof}
	If we show that $A_1\otimes A_2$ is $\tilde{\cK}_{\max}$-positive, 
	the assertion follows immediately. 
	Instead of the $\tilde{\cK}_{\max}$-positivity of $A_1\otimes A_2$, 
	we may show that $(A_1\otimes A_2)^\ast$ is $\tilde{\cK}_{\max}^\ast$-positive, 
	thanks to Proposition~\ref{dual-pos}. 
	It can be easily checked that the following three facts hold: 
	$A_1^\ast\otimes A_2^\ast$ is $\cK_1^\ast\otimes\cK_2^\ast$-positive; 
	$A_1^\ast\otimes A_2^\ast = (A_1\otimes A_2)^\ast$; 
	$\tilde{\cK}_{\max}^\ast = \cK_1^\ast\otimes\cK_2^\ast$. 
	Therefore, $(A_1\otimes A_2)^\ast$ is $\tilde{\cK}_{\max}^\ast$-positive.
\end{proof}

\begin{lemma}[$(\cK,u)$-norm]\label{lem:norm}
	For $x\in\cV$, the value $\|x\|_{\cK,u}$ is defined as 
	\[
	\|x\|_{\cK,u} = \max_{-u\le_{\cK^\ast} y\le_{\cK^\ast} u} \abs{\braket{y,x}},
	\]
	where the partial order $y\le_{\cK^\ast} y'$ denotes $y'-y\in\cK^\ast$. 
	Then $\|\cdot\|_{\cK,u}$ is a norm on $\cV$, 
	which we call $(\cK,u)$-norm. 
	If $A$ is a $(\cK,u; \cK',u')$-dynamical map, 
	then any $x\in\cV$ satisfies $\|Ax\|_{\cK',u'} \le \|x\|_{\cK,u}$.
\end{lemma}
\begin{proof}
	First, we show that $\|\cdot\|_{\cK,u}$ is a norm. 
	From the definition, the absolute homogeneity and the triangle inequality follow immediately. 
	Hence all we need is to show that the equation $\|x\|_u=0$ implies $x=0$. 
	Thanks to $\|x\|_{\cK,u}=0$, 
	any $-u\le_{\cK^\ast} y\le_{\cK^\ast} u$ satisfies $\braket{y,x}=0$. 
	Take an arbitrary $y\in\cV$. 
	Since the unit effect $u$ is an interior point of $\cK^\ast$, 
	the inequality $-u\le_{\cK^\ast} \epsilon y\le_{\cK^\ast} u$ holds for a sufficiently small $\epsilon>0$. 
	Thus $\epsilon\braket{y,x}=0$ and $\braket{y,x}=0$. 
	Since $y\in\cV$ is arbitrary, we obtain $x=0$. 
	Therefore, $\|\cdot\|_u$ is a norm.
	\par
	Next, in order to show the remaining of Lemma~\ref{lem:norm}, 
	let $A$ be a $(\cK,u; \cK',u')$-dynamical map. 
	The inequality $-u'\le_{(\cK')^\ast} y'\le_{(\cK')^\ast} u'$ and Proposition~\ref{dual-pos} 
	yield 
	\[
	-u = -A^\ast u' \le_{\cK^\ast} A^\ast y \le_{\cK^\ast} A^\ast u' = u.
	\]
	Therefore, 
	\begin{align*}
		\|Ax\|_{\cK',u'} &= \max_{-u'\le_{(\cK')^\ast} y'\le_{(\cK')^\ast} u'} \abs{\braket{y',Ax}}
		= \max_{-u'\le_{(\cK')^\ast} y'\le_{(\cK')^\ast} u'} \abs{\braket{A^\ast y',x}}\\
		&\le \max_{-u\le_{\cK^\ast} y\le_{\cK^\ast} u} \abs{\braket{y,x}} = \|x\|_{\cK,u}.
	\end{align*}
\end{proof}

\begin{proof}[Proof of Lemma~$\ref{map-compact}$]
	From the definition, 
	it follows that the set of all $(\cK,u; \cK',u')$-dynamical maps is closed. 
	Thus, to prove the compactness, 
	all we need is to show that the set of all $(\cK,u; \cK',u')$-dynamical maps is bounded. 
	Let $A$ be a $(\cK,u; \cK',u')$-dynamical map. 
	Consider the operator norm $\|A\|_{\cK,u; \cK',u'}$ 
	based on the norms $\|\cdot\|_{\cK,u}$ and $\|\cdot\|_{\cK',u'}$: 
	\[
	\|A\|_{\cK,u; \cK',u'} \coloneqq \min_{\|x\|_{\cK,u}\le1} \|Ax\|_{\cK',u'}.
	\]
	We show $\|A\|_{\cK,u; \cK',u'} \le 1$. 
	Since $A$ is $(\cK,u; \cK',u')$-dynamical, 
	Proposition~\ref{dual-pos} implies that 
	any $y'\in\cV'$ with $-u'\le_{(\cK')^\ast} y'\le_{(\cK')^\ast} u'$ satisfies 
	\[
	-u = -A^\ast u' \le_{\cK^\ast} A^\ast y' \le_{\cK^\ast} A^\ast u' = u.
	\]
	Therefore, any $x\in\cV$ satisfies 
	\begin{align*}
		\|Ax\|_{\cK',u'}
		&= \max_{-u'\le_{(\cK')^\ast} y'\le_{(\cK')^\ast} u'} \abs{\braket{y',Ax}}
		= \max_{-u'\le_{(\cK')^\ast} y'\le_{(\cK')^\ast} u'} \abs{\braket{A^\ast y',x}}\\
		&\le \max_{-u\le_{\cK^\ast} y\le_{\cK^\ast} u} \abs{\braket{y,x}}
		= \|x\|_{\cK,u},
	\end{align*}
	whence $\|A\|_{\cK,u; \cK',u'}\le1$. 
	Since $A$ is arbitrary, we conclude the boundedness.
\end{proof}

\begin{lemma}[Partial diagonalizability]\label{p-diag}
	Let $A$ be a dynamical map and $\lambda$ be an eigenvalue of $A$ whose absolute value is $r(A)=1$. 
	Then $\am(A; \lambda)=\gm(A; \lambda)$.
\end{lemma}
\begin{proof}
	For $\lambda\in\mathbb{C}$ and $s\in\mathbb{N}$, 
	we denote by $J_s(\lambda)$ the Jordan block 
	\[
	\begin{bmatrix}
		\lambda&1&   &   \\
		   &\lambda&\ddots&   \\
		   &   &\ddots&1\\
		   &   &   &\lambda
	\end{bmatrix}
	\]
	whose size is $s$. 
	Consider the Jordan canonical form of $A$: 
	there exists an invertible linear map $T$ such that 
	\[
	T^{-1}AT = J \coloneqq
	\begin{bmatrix}
		J_{s_1}(\lambda_1)&   &   \\
		   &\ddots&   \\
		   &   &J_{s_r}(\lambda_r)
	\end{bmatrix},
	\]
	where $\lambda_1,\ldots,\lambda_r$ are eigenvalues of $A$. 
	We show the assertion by contradiction. 
	Suppose that there exists an eigenvalue $\lambda$ of $A$ 
	such that $|\lambda|=1$ and $\gm(A; \lambda) < \am(A; \lambda)$. 
	Then $s_{i_0}\ge2$ and $\lambda_{i_0}=\lambda$ for some integer $1\le i_0\le r$. 
	This fact implies that $J_2(\lambda)$ is a submatrix of $J_{s_{i_0}}(\lambda_{i_0})$ 
	composed of the $(1,1)$, $(1,2)$, $(2,1)$, and $(2,2)$ entries. 
	Therefore, we have 
	\begin{gather*}
		\|J_{s_{i_0}}(\lambda_{i_0})^n\|_2^2 \ge \|J_2(\lambda)^n\|_2^2
		= |\lambda^n|^2 + |\lambda^n|^2 + |n\lambda^{n-1}|^2 > n^2,\\
		\begin{split}
			\|T^{-1}\|_2 \|A^n\|_2 \|T\|_2 &\ge \|T^{-1}A^n T\|_2
			= \|J^n\|_2 = \Bigl( \sum_{i=1}^r \|J_{s_i}(\lambda_i)^n\|_2^2 \Bigr)^{1/2}\\
			&\ge \|J_{s_{i_0}}(\lambda_{i_0})^n\|_2 > n \to \infty,
		\end{split}
	\end{gather*}
	where $\|X\|_2 \coloneqq (\Tr X^\ast X)^{1/2}$ for a matrix $X$.  
	However, the second inequality contradicts Lemma~\ref{map-compact}.
\end{proof}

\begin{proof}[Proof of Lemma~$\ref{inv-ran}$]
	Consider the Jordan canonical form $J=T^{-1}AT$ of $A$, 
	where $T$ is an invertible linear map. 
	Thanks to Lemma~\ref{p-diag}, 
	the Jordan canonical form $J$ can be expressed as 
	\[
	J =
	\begin{bmatrix}
		D&   &   &   \\
		   &J_{s_1}(\lambda'_1)&   &   \\
		   &   &\ddots&   \\
		   &   &   &J_{s_r}(\lambda'_{r'})
	\end{bmatrix}
	,\quad D =
	\begin{bmatrix}
		\lambda_1&   &   \\
		   &\ddots&   \\
		   &   &\lambda_r
	\end{bmatrix},
	\]
	where $|\lambda_i|=1$ and $|\lambda'_{i'}|<1$ for all $1\le i\le r$ and $1\le i'\le r'$. 
	Thus 
	\begin{align*}
		A_0 &= \lim_{N\ni n\to\infty} A^n
		= \lim_{N\ni n\to\infty} T J^n T^{-1}\\
		&= \lim_{N\ni n\to\infty} T
		\begin{bmatrix}
			D^n&   &   &   \\
			   &J_{s_1}(\lambda'_1)^n&   &   \\
			   &   &\ddots&   \\
			   &   &   &J_{s_r}(\lambda'_{r'})^n
		\end{bmatrix}
		T^{-1}
		= T
		\begin{bmatrix}
			D_0&   &   &   \\
			   &0&   &   \\
			   &   &\ddots&   \\
			   &   &   &0\\
		\end{bmatrix}
		T^{-1},
	\end{align*}
	where $D_0$ is a diagonal unitary matrix whose size is $r$. 
	Similarly, for some diagonal unitary matrix $D'_0$ whose size is $r$, 
	we have 
	\[
	A'_0 = \lim_{N'\ni n\to\infty} A^n
	= T
	\begin{bmatrix}
		D'_0&   &   &   \\
		   &0&   &   \\
		   &   &\ddots&   \\
		   &   &   &0\\
	\end{bmatrix}
	T^{-1}.
	\]
	Therefore, $\Ran A_0=\Ran A'_0$.
\end{proof}

\begin{proof}[Proof of Lemma~$\ref{lem:ad}$]
	Take an arbitrary state $\tilde{x}\in\cS(\tilde{\cK}, u'_1\otimes u'_2)$, 
	and put 
	\begin{gather*}
		c_n \coloneqq \Bigl( \frac{\braket{u'_1\otimes u'_2, \tilde{A}^n \tilde{x}}}
		{\braket{u_1\otimes u_2, \tilde{A}^n \tilde{x}}} \Bigr)^{1/2},\quad
		x_{n,i} \coloneqq c_n^{-1}\pi_i\Bigl( \frac{\tilde{A}^n \tilde{x}}
		{\braket{u_1\otimes u_2, \tilde{A}^n \tilde{x}}}; u_1,u_2 \Bigr)
		\quad(i=1,2).
	\end{gather*}
	Since $\tilde{A}$ is $(\tilde{\cK}, u_1, u_2)$-asymptotically decoupling, 
	we have 
	\begin{align}
		&\quad c_n^2 \Bigl( \frac{\tilde{A}^n \tilde{x}}{\braket{u'_1\otimes u'_2, \tilde{A}^n \tilde{x}}}
		- x_{n,1}\otimes x_{n,2} \Bigr)\nonumber\\
		&= \frac{\tilde{A}^n \tilde{x}}{\braket{u_1\otimes u_2, \tilde{A}^n \tilde{x}}}
		- \pi_1\Bigl( \frac{\tilde{A}^n \tilde{x}}{\braket{u_1\otimes u_2, \tilde{A}^n \tilde{x}}}; u_1,u_2 \Bigr)\otimes
		\pi_2\Bigl( \frac{\tilde{A}^n \tilde{x}}{\braket{u_1\otimes u_2, \tilde{A}^n \tilde{x}}}; u_1,u_2 \Bigr)\nonumber\\
		&\xrightarrow{n\to\infty} 0. \label{eq05}
	\end{align}
	Since $\tilde{A}^n x/\braket{u_1\otimes u_2,\tilde{A}^n x}$ is a state 
	in $\cS(\tilde{\cK}, u_1\otimes u_2)$ 
	and the set $\cS(\tilde{\cK},u_1\otimes u_2)$ is compact, 
	the sequence $\{c_n\}_n$ of positive numbers is bounded. 
	Thus \eqref{eq05} implies 
	\begin{align}
		\frac{\tilde{A}^n \tilde{x}}{\braket{u'_1\otimes u'_2,\tilde{A}^n \tilde{x}}}
		&= x_{n,1}\otimes x_{n,2} + o(1), \label{eq06}\\
		\braket{u'_1,x_{n,1}}\braket{u'_2,x_{n,2}}
		&= \braket{u'_1\otimes u'_2, x_{n,1}\otimes x_{n,2}}\nonumber\\
		&= \biggl\langle u'_1\otimes u'_2,\ \frac{\tilde{A}^n \tilde{x}}
		{\braket{u'_1\otimes u'_2,\tilde{A}^n \tilde{x}}} \biggr\rangle
		+ o(1)\nonumber\\
		&= 1 + o(1). \label{eq07}
	\end{align}
	Noting that the sequence $\{x_{n,1}\otimes x_{n,2}\}_n$ is bounded, 
	we obtain 
	\begin{align*}
		&\quad \frac{\tilde{A}^n \tilde{x}}{\braket{u'_1\otimes u'_2,\tilde{A}^n \tilde{x}}}
		- \pi_1\Bigl( \frac{\tilde{A}^n \tilde{x}}{\braket{u'_1\otimes u'_2,\tilde{A}^n \tilde{x}}}; u'_1,u'_2 \Bigr)\otimes
		\pi_2\Bigl( \frac{\tilde{A}^n \tilde{x}}{\braket{u'_1\otimes u'_2,\tilde{A}^n \tilde{x}}}; u'_1,u'_2 \Bigr)\\
		&\overset{(a)}{=} x_{n,1}\otimes x_{n,2}
		- \pi_1(x_{n,1}\otimes x_{n,2}; u'_1,u'_2)\otimes\pi_2(x_{n,1}\otimes x_{n,2}; u'_1,u'_2) + o(1)\\
		&= x_{n,1}\otimes x_{n,2}
		- \braket{u'_1,x_{n,1}}\braket{u'_2,x_{n,2}} x_{n,1}\otimes x_{n,2} + o(1)\\
		&= (1 - \braket{u'_1,x_{n,1}}\braket{u'_2,x_{n,2}}) x_{n,1}\otimes x_{n,2} + o(1)
		\overset{(b)}{=} o(1),
	\end{align*}
	where \eqref{eq06} and \eqref{eq07} have been used to obtain $(a)$ and $(b)$, respectively. 
	Therefore, $\tilde{A}$ is $(\tilde{\cK}, u'_1, u'_2)$-asymptotically decoupling.
\end{proof}

\begin{proof}[Proof of Proposition~$\ref{prop:ergL}$]
	Without loss of generality, we may assume that 
	$A$ is a Jordan canonical form and satisfies $r(A)=1$. 
	Then the convergence of $(1/n)\sum_{k=0}^{n-1} A^k$ can be reduced to 
	that of $(1/n)\sum_{k=0}^{n-1} J_s(\lambda)^k$ for each Jordan block $J_s(\lambda)$ of $A$.
	\par
	\eqref{ergL2}$\Rightarrow$\eqref{ergL1}. 
	Let us investigate the convergence of $(1/n)\sum_{k=0}^{n-1} J_s(\lambda)^k$. 
	Condition~\eqref{ergL2} implies that any Jordan block $J_s(\lambda)$ of $A$ satisfies 
	either $|\lambda|<1$ or $(s,|\lambda|)=(1,1)$. 
	First, if $|\lambda|<1$, we have $(1/n)\sum_{k=0}^{n-1} J_s(\lambda)^k \to 0$ as $n\to\infty$. 
	Second, if $(s,|\lambda|)=(1,1)$ and $\lambda\not=1$, we have 
	\[
	\frac{1}{n}\sum_{k=0}^{n-1} J_s(\lambda)^k
	= \Bigl[ \frac{1}{n}\sum_{k=0}^{n-1} \lambda^k \Bigr]
	= \Bigl[ \frac{1}{n}\frac{1-\lambda^n}{1-\lambda} \Bigr] \xrightarrow{n\to\infty} 0.
	\]
	Third, if $(s,\lambda)=(1,1)$, 
	any $n\in\mathbb{N}$ satisfies $(1/n)\sum_{k=0}^{n-1} J_s(\lambda)^k = [1]$. 
	From the above three cases and $\am(A; 1)=1$, it follows that $A$ is ergodic.
	\par
	\eqref{ergL1}$\Rightarrow$\eqref{ergL2}. 
	First, we show that 
	there is not a Jordan block $J_s(\lambda)$ of $A$ satisfying $s\ge2$ and $|\lambda|=1$ 
	by contradiction. 
	Suppose that there is a Jordan block $J_2(\lambda)$ of $A$ satisfying $|\lambda|=1$. 
	If $\lambda\not=1$, the equation 
	\[
	\frac{1}{n}\sum_{k=0}^{n-1} J_2(\lambda)^k 
	= \frac{1}{n}\sum_{k=0}^{n-1}
	\begin{bmatrix}
		\lambda&1\\
		0&\lambda
	\end{bmatrix}
	^k = \frac{1}{n}\sum_{k=0}^{n-1}
	\begin{bmatrix}
		\lambda^k&k\lambda^{k-1}\\
		0&\lambda^k
	\end{bmatrix}
	=
	\begin{bmatrix}
		o(1)&(1/n)\sum_{k=0}^{n-1} k\lambda^{k-1}\\
		0&o(1)
	\end{bmatrix}
	\]
	holds. Using the formula 
	\[
	\sum_{k=0}^{n-1} k\lambda^{k-1} = \frac{1-\lambda^n}{(1-\lambda)^2} - \frac{n\lambda^{n-1}}{1-\lambda},
	\]
	we obtain $(1/n)\sum_{k=0}^{n-1} k\lambda^{k-1} = o(1)+\lambda^{n-1}/(1-\lambda)$. 
	However, it contradicts condition~\eqref{ergL1}, since $\lambda^{n-1}/(1-\lambda)$ does not converge. 
	Similarly, if $\lambda=1$, the equation 
	\[
	\frac{1}{n}\sum_{k=0}^{n-1} J_2(\lambda)^k 
	= \frac{1}{n}\sum_{k=0}^{n-1}
	\begin{bmatrix}
		1&1\\
		0&1
	\end{bmatrix}
	^k = \frac{1}{n}\sum_{k=0}^{n-1}
	\begin{bmatrix}
		1&k\\
		0&1
	\end{bmatrix}
	=
	\begin{bmatrix}
		1&(n-1)/2\\
		0&1
	\end{bmatrix}
	\]
	holds, which contradicts condition~\eqref{ergL1}. 
	Therefore, there is not a Jordan block $J_2(\lambda)$ of $A$ satisfying $|\lambda|=1$. 
	It can be proved that 
	there is not a Jordan block $J_s(\lambda)$ of $A$ satisfying $s\ge2$ and $|\lambda|=1$ 
	in the same way. 
	This result can be rewritten as 
	that $\am(A; \lambda)=\gm(A; \lambda)$ for any $\lambda\in\mathbb{C}$ whose absolute value is one.
	\par
	Next, we show that $A$ has the eigenvalue $1$ by contradiction. 
	Suppose that $A$ did not have the eigenvalue $1$. 
	Since any Jordan block $J_s(\lambda)$ satisfying $|\lambda|=1$ must be $J_1(\lambda)$ 
	due to the result in the previous paragraph, 
	any Jordan block $J_s(\lambda)$ of $A$ satisfies 
	$(1/n)\sum_{k=0}^{n-1} J_s(\lambda)^k \to 0$. 
	Therefore, $(1/n)\sum_{k=0}^{n-1} A^k \to 0$, which contradicts condition~\eqref{ergL1}. 
	This contradiction implies that $A$ has the eigenvalue $1$.
	\par
	Since we have already shown $\am(A; 1)=\gm(A; 1)$, 
	the remaining of condition~\eqref{ergL2} is $\gm(A; 1)=1$. 
	It follows from condition~\eqref{ergL1} immediately.
\end{proof}

\section{Proof of Proposition~\ref{prop:irrP}}\label{appB}
In this appendix, we prove Proposition~\ref{prop:irrP}. 
For $x,x'\in\cV$, 
we define $x\le_{\cK} x'$ and $x<_{\cK}x'$ as $x'-x\in\cK$ and $x'-x\in\tint(\cK)$, respectively.

\begin{lemma}\label{lem:B2}
	If $A$ is an ergodic map with a stationary vector $x_0$ 
	and a dual stationary vector $y_0$, 
	then $A^\ast$ is ergodic and has the stationary vector $y_0$ 
	and the dual stationary vector $x_0$.
\end{lemma}
\begin{proof}
	Without loss of generality, we may assume $r(A)=1$. 
	Then ergodicity for $A$ is equivalent to the statement that any $x,y\in\cV$ satisfy 
	\begin{equation*}
		\lim_{n\to\infty} \frac{1}{n}\sum_{k=0}^{n-1} \braket{y,A^k x} = \braket{y_0,x}\braket{y,x_0}.
	\end{equation*}
	This statement is also equivalent to ergodicity for $A^\ast$.
\end{proof}

\begin{lemma}\label{lem:B3}
	If $A$ is an ergodic map with a stationary vector $x_0$ 
	and a dual stationary vector $y_0$, 
	then $x_0$ and $y_0$ are eigenvectors of $A$ and $A^\ast$ 
	associated with $r(A)$, respectively.
\end{lemma}
\begin{proof}
	Without loss of generality, we may assume $r(A)=1$. 
	Thanks to Proposition~\ref{prop:ergL}, 
	we can take an eigenvector $x_1$ of $A$ associated with $r(A)=1$. 
	The ergodicity of $A$ yields 
	\[
	x_1 = \lim_{n\to\infty} \frac{1}{n}\sum_{k=0}^{n-1} A^k x_1
	= \braket{y_0,x_1}x_0.
	\]
	Since $x_1\not=0$, 
	the stationary vector $x_0$ is an eigenvector of $A$ associated with $r(A)=1$. 
	Also, $y_0$ is an eigenvector of $A^\ast$ associated with $r(A)=1$ 
	thanks to Lemma~\ref{lem:B2}.
\end{proof}

\begin{lemma}\label{lem:B1}
	If a $\cK$-positive map $A$ has 
	an eigenvector $x_0\in\tint(\cK)$ associated with $0$, 
	then $A=0$.
\end{lemma}
\begin{proof}
	Take an arbitrary $y\in\cK^\ast$. 
	Then 
	\[
	\braket{A^\ast y, x_0} = \braket{y, Ax_0} = 0.
	\]
	Also, $x_0\in\tint(\cK)$ and $A^\ast y\in\cK^\ast$ because of Proposition~\ref{dual-pos}. 
	The above relations and Lemma~\ref{lem:appC1} in Appendix~\ref{appC} imply $A^\ast y=0$. 
	Moreover, the relation 
	$\cV = \cK^\ast + (-\cK^\ast) \coloneqq \set{y-y' | y,y'\in\cK^\ast}$ 
	holds thanks to Lemma~\ref{cone-span} in Appendix~\ref{appC}. 
	From this relation and the arbitrariness of $y\in\cK^\ast$, 
	it follows that $A^\ast=0$ and $A=0$.
\end{proof}

\begin{proof}[Proof of Proposition~$\ref{prop:irrP}$]
	\eqref{evec2}$\Rightarrow$\eqref{mean}. 
	Assume condition~\eqref{evec2}. 
	Since $r(A)^{-1}A$ is $(\cK,y_0)$-dynamical, 
	Proposition~\ref{prop:ergD} and condition~\ref{evec2} imply that $A$ is ergodic. 
	From Lemma~\ref{lem:B3} and Proposition~\ref{prop:ergP}, 
	it follows that $x_0$ and $y_0$ are a stationary vector 
	and a dual stationary vector, respectively. 
	Therefore, $A$ is $\cK$-irreducible.
	\par
	\eqref{mean}$\Rightarrow$\eqref{exp}. 
	Assume condition~\eqref{mean}. 
	Let $x_0\in\tint(\cK)$ and $y_0\in\tint(\cK^\ast)$ 
	be a stationary vector and a dual stationary vector, respectively. 
	Take an arbitrary $x\in\cK\setminus\{ 0 \}$. 
	Since $(1/n)\sum_{k=0}^{n-1} (r(A)^{-1}A)^k x$ converges to $\braket{y_0,x}x_0\in\tint(\cK)$, 
	a sufficiently large $n\in\mathbb{N}$ satisfies $\sum_{k=0}^n (r(A)^{-1}A)^k x\in\tint(\cK)$. 
	Thus 
	\[
	e^{r(A)^{-1}A}x \ge_{\cK} \sum_{k=0}^n \frac{r(A)^{-k}}{k!}A^k x
	\ge_{\cK} \frac{1}{n!}\sum_{k=0}^n (r(A)^{-1}A)^k x
	>_{\cK} 0,
	\]
	which is just condition~\eqref{exp}.
	\par
	\eqref{exp}$\Rightarrow$\eqref{improve}. 
	Assume condition~\eqref{exp}. 
	For $x\in\cV$ and $n\in\mathbb{N}$, 
	define $r_n(x)$ to be the dimension of $\vspan(x,Ax,\ldots,A^{n-1}x)$. 
	Let $n_0(x)$ be the minimum number $n\in\mathbb{N}$ satisfying $r_n(x)=r_{n+1}(x)$. 
	Since the dimension of $\cV$ is $d$, 
	any $x\in\cV$ satisfies $1\le n_0(x)\le d$. 
	By induction, it can be easily checked that any integer $n\ge n_0(x)$ satisfies $r_n(x)=r_{n+1}(x)$. 
	In particular, $r_d(x)=r_{d+1}(x)=\cdots$ holds. 
	Now, let us show condition~\eqref{improve} by contradiction. 
	Suppose that the boundary of $\cK$ contained $(I+A)^{d-1}x_1$ 
	for some $x_1\in\cK\setminus\{ 0 \}$. 
	Then there exists $y_1\in\cK^\ast\setminus\{ 0 \}$ such that $\braket{y_1,(I+A)^{d-1}x_1}=0$. 
	From the expansion $(I+A)^{d-1}=\sum_{k=0}^{d-1} \binom{d-1}{k}A^k$, 
	any integer $0\le k\le d-1$ satisfies $\braket{y_1,A^k x_1}=0$. 
	For any integer $k\ge0$, 
	since the equation $r_d(x_1)=r_{d+1}(x_1)=\cdots$ implies 
	$A^k x_1\in\vspan(x_1, Ax_1,\ldots, A^{d-1}x_1)$, 
	we obtain $\braket{y_1,A^k x_1}=0$. 
	Thus any $t>0$ satisfies $\braket{y_1,e^{tA}x_1} = \sum_{k=0}^\infty t^k\braket{y_1,A^k x_1}/k! = 0$, 
	which contradicts condition~\eqref{exp}.
	\par
	\eqref{improve}$\Rightarrow$\eqref{ineq}. 
	Assume condition~\eqref{improve}. 
	Let $x\in\cK\setminus\{ 0 \}$, $\alpha\ge0$, and $Ax\le_{\cK}\alpha x$. 
	Then the inequality $0 <_{\cK} (I+A)^{d-1}x \le_{\cK} (1+\alpha)^{d-1}x$ yields $x>_{\cK}0$.
	\par
	\eqref{ineq}$\Rightarrow$\eqref{evec1}. 
	Assume condition~\eqref{ineq}. 
	Suppose that $Ax=\lambda x$ for some $\lambda\in\mathbb{R}$ and $x\in\cK\setminus\{ 0 \}$. 
	Then $Ax = \lambda x \le_{\cK} \abs{\lambda}x$, 
	and condition~\eqref{ineq} implies $x\in\tint(\cK)$. 
	Therefore, condition~\eqref{evec1} holds.
	\par
	\eqref{evec1}$\Rightarrow$\eqref{evec2}. 
	Assume condition~\eqref{evec1}. 
	Thanks to Proposition~\ref{prop:pos} and Proposition~\ref{dual-pos}, 
	the linear maps $A$ and $A^\ast$ have eigenvectors $x_0\in\cK$ 
	and $y_0\in\cK^\ast$ associated with $r(A)$, respectively. 
	Condition~\eqref{evec1} implies $x_0\in\tint(\cK)$. 
	First, we show $r(A)>0$ by contradiction. 
	Suppose $r(A)=0$, Lemma~\ref{lem:B1} yields $A=0$, 
	which contradicts condition~\eqref{evec1} because of $\dim\cV\ge2$. 
	Therefore, $r(A)>0$.
	\par
	Second, we show $\gm(A; r(A))=1$. 
	Let $x$ be an eigenvector of $A$ associated with $r(A)$. 
	From Lemma~\ref{lem:appC3} in Appendix~\ref{appC}, 
	we can take a real number $\alpha\not=0$ 
	such that the boundary of $\cK$ contains $x' \coloneqq x_0 - \alpha x$. 
	Since $Ax'=r(A)x'$, condition~\eqref{evec1} implies $x'=0$, 
	whence $x\in\vspan(x_0)$. 
	Therefore, $\gm(A; r(A))=1$.
	\par
	Third, we prove $y_0\in\tint(\cK^\ast)$ by contradiction. 
	Suppose that $y_0\in\cK^\ast\setminus\{ 0 \}$ belonged to the boundary of $\cK^\ast$. 
	Defining 
	\[
	\cK_0 = \set{x\in\cK | \braket{y_0,x}=0},\quad 
	\cV_0 = \vspan(\cK_0),
	\]
	we find that $\dim\cV_0\ge1$ and $A\cK_0\subset\cK_0$. 
	Since $\cK_0\subset\cV_0$ is a proper cone, 
	Proposition~\ref{prop:pos} implies that 
	the restriction of $A$ to $\cV_0$ has an eigenvector in $\cK_0$. 
	However, the boundary of $\cK$ contains $\cK_0$, 
	which contradicts condition~\eqref{evec1}. 
	Therefore, $y_0\in\tint(\cK^\ast)$. 
	The results above and in the previous paragraphs are just condition~\eqref{evec2}.
	\par
	\eqref{mean}$\Rightarrow$\eqref{exp}$'$. 
	Assume condition~\eqref{mean}. 
	Let $x$ and $y$ be nonzero extreme vectors of $\cK$ and $\cK^\ast$, respectively. 
	Thanks to condition~\eqref{mean}, 
	the map $A$ has a stationary vector $x_0\in\tint(\cK)$ 
	and a dual stationary vector $y_0\in\tint(\cK^\ast)$. 
	Since the ergodicity of $A$ implies 
	\[
	\lim_{n\to\infty} \frac{1}{n}\sum_{k=1}^{n-1} (r(A)^{-1}A)^k x
	= \lim_{n\to\infty} \frac{1}{n}\sum_{k=0}^{n-1} (r(A)^{-1}A)^k x
	= \braket{y_0,x}x_0 \in \tint(\cK),
	\]
	a sufficiently large $n\in\mathbb{N}$ satisfies 
	$\sum_{k=1}^n (r(A)^{-1}A)^k x\in\tint(\cK)$. 
	Therefore, 
	\[
	0 < \Bigl\langle y,\ \sum_{k=1}^n (r(A)^{-1}A)^k x \Bigr\rangle
	= \sum_{k=1}^n r(A)^{-k}\braket{y, A^k x},
	\]
	whence $\braket{y, A^k x}>0$ for some integer $1\le k\le n$.
	\par
	\eqref{exp}$'$$\Rightarrow$\eqref{exp}. 
	Assume condition~\eqref{exp}$'$. 
	Let $x$ be a nonzero extreme vector of $\cK$. 
	Since for any nonzero extreme $y\in\cK^\ast$ 
	there exists $n\in\mathbb{N}$ such that $\braket{y,A^n x}>0$, 
	any nonzero extreme $y\in\cK^\ast$ satisfies 
	$\braket{y,e^A x} \ge (1/n!)\braket{y,A^n x} > 0$. 
	Since $\cK^\ast$ is generated by extreme vectors of $\cK^\ast$, 
	(which is proved by using the Krein-Milman theorem; for details, see \cite{Vander}), 
	any $y\in\cK^\ast\setminus\{ 0 \}$ satisfies $\braket{y,e^A x}>0$, 
	which implies $e^A x\in\tint(\cK)$ due to Lemma~\ref{lem:appC1} in Appendix~\ref{appC}. 
	Since $\cK$ is generated by extreme vectors of $\cK$, 
	any $x\in\cK\setminus\{ 0 \}$ satisfies $e^A x\in\tint(\cK)$.
\end{proof}

\section{Basic lemmas on convex cones}\label{appC}
For readers' convenience, we prove basic lemmas on convex cones. 
The lemmas are often used in our main discussion implicitly/explicitly. 
Recall that $\cK$ is called a closed convex cone (for short, cone) if $\cK$ is a closed convex cone, 
and $\cK$ is called a proper cone 
if $\cK$ is a cone satisfying $\tint(\cK) \not= \emptyset$ and $\cK\cap(-\cK) = \{ 0 \}$.

\begin{lemma}\label{lem:appC1}
	For any cone $\cK\not=\emptyset$, the following relations hold: 
	\begin{align*}
		\partial(\cK^\ast) &= \set{y\in\cV | \braket{y,x}=0\ (\exists x\in\cK\setminus\{ 0 \})},\\
		\partial\cK &= \set{x\in\cV | \braket{y,x}=0\ (\exists y\in\cK^\ast\setminus\{ 0 \})},\\
		\tint(\cK^\ast) &= \set{y\in\cV | \braket{y,x}>0\ (\forall x\in\cK\setminus\{ 0 \})},\\
		\tint(\cK) &= \set{x\in\cV | \braket{y,x}>0\ (\forall y\in\cK^\ast\setminus\{ 0 \})},
	\end{align*}
	where $\partial\cX$ denotes the boundary of a set $\cX\subset\cV$.
\end{lemma}
\begin{proof}
	Let $\cS$ be the unit sphere of $\cV$. 
	We show 
	\begin{equation}
		\tint(\cK^\ast) = \set{y\in\cV | \braket{y,x}>0\ (\forall x\in\cK\cap\cS)}.
		\label{eq08}
	\end{equation}
	Suppose that some $y\in\cV$ satisfies that $\braket{y,x}>0$ for any $x\in\cK\cap\cS$. 
	Then $y\not=0$. 
	Any $y'\in\cV$ with $\|y'\| < \min_{x\in\cK\cap\cS} \braket{y,x}$ satisfies 
	\[
	\min_{x\in\cK\cap\cS} \braket{y+y',x}
	\ge \min_{x\in\cK\cap\cS} \braket{y,x} - \|y'\| > 0.
	\]
	This inequality implies that $y+y'\in\cK^\ast$ 
	whenever $\|y'\| < \min_{x\in\cK\cap\cS} \braket{y,x}$. 
	Thus $y\in\tint(\cK^\ast)$, 
	which concludes that $\tint(\cK^\ast)$ contains the right-hand side of \eqref{eq08}.
	\par
	Next, we prove the remaining part of \eqref{eq08} by contradiction. 
	Take $y\in\tint(\cK^\ast)$. 
	Suppose that $\braket{y,x}=0$ for some $x\in\cK\cap\cS$. 
	Since $y\in\tint(\cK^\ast)$, 
	we can take a sufficiently small $\epsilon>0$ such that $y - \epsilon x\in\cK^\ast$. 
	However, the inequality $0 \le \braket{y - \epsilon x,x} = -\epsilon < 0$ is a contradiction. 
	Thus the relation \eqref{eq08} holds. 
	Since $\cK$ is a cone, we obtain the third relation.
	\par
	Recall that for any cone $\cK\not=\emptyset$ the relation $\cK^{\ast\ast}=\cK$ holds. 
	Using this relation and the third relation, we obtain the fourth relation. 
	The first and second relations follow from the third and fourth ones, respectively.
\end{proof}

\begin{lemma}\label{cone-span}
	Let $\cK$ be a convex cone having non-empty interior. 
	Then $\cV = \cK + (-\cK) \coloneqq \set{x_1 - x_2 | x_1,x_2\in\cK}$.
\end{lemma}
\begin{proof}
	From the definition, it follows that $\cK + (-\cK)$ is a linear subspace of $\cV$. 
	Since $\cK$ has non-empty interior, so does $\cK + (-\cK)$. 
	Therefore, $\cV = \cK + (-\cK)$.
\end{proof}

\begin{lemma}\label{lem:appC2}
	Let $\cK$ be a proper cone. 
	Then there exists a unit vector $e\in\tint(\cK)$ such that 
	\[
	\cK\setminus\{ 0 \} \subset \set{x\in\cV | \braket{e,x}>0}.
	\]
	In particular, $\tint(\cK)\cap\tint(\cK^\ast)\not=\emptyset$.
\end{lemma}
\begin{proof}
	Let $\cS$ be the unit sphere of $\cV$ and $\cl(\cB)$ be the closed unit ball of $\cV$, 
	and define $\cK_0 = \conv(\cK\cap\cS) \subset\cK\cap\cl(\cB)$, 
	where the symbol $\conv(\cX)$ denotes the convex hull of a subset $\cX\subset\cV$. 
	Then $\cK_0$ is a compact convex set. 
	The zero vector is an extreme point of $\cK$ owing to $\cK\cap(-\cK) = \{ 0 \}$, 
	whence the relation $0\not\in\cK\cap\cS$ implies $0\not\in\cK_0$. 
	Therefore, 
	\begin{align*}
		&\sup_{y\in\tint(\cK)\cap\cl(\cB)} \min_{x\in\cK\cap\cS} \braket{y,x}
		\overset{(a)}{=} \max_{y\in\cK\cap\cl(\cB)} \min_{x\in\cK\cap\cS} \braket{y,x}
		\overset{(b)}{=} \max_{y\in\cK\cap\cl(\cB)} \min_{x\in\cK_0} \braket{y,x}\\
		\ge& \max_{y\in\cK_0} \min_{x\in\cK_0} \braket{y,x}
		\overset{(c)}{=} \min_{x\in\cK_0} \max_{y\in\cK_0} \braket{y,x} \ge \min_{x\in\cK_0} \|x\|^2 > 0.
	\end{align*}
	This inequality implies the assertion.
	\par
	We verify the above inequality. 
	First, the equality $(b)$ holds because the function $\braket{y,x}$ of $x\in\cK_0$ 
	achieves the minimum number when $x$ is an extreme point of $\cK_0$. 
	Second, the equality $(c)$ follows from the minimax theorem. 
	Third, the equality $(a)$ can be shown as follows. 
	Define $f(y) = \min_{x\in\cK\cap\cS} \braket{y,x}$ for $y\in\cV$. 
	Then the function $f$ is concave. 
	Since $\tint(\cK)\not=\emptyset$, 
	we can take $y_0\in\tint(\cK)\cap\cl(\cB)$. 
	Thus any $y_1\in\partial\cK\cap\cl(\cB)$ and $0\le t<1$ satisfy 
	$y_t \coloneqq (1-t)y_0 + ty_1\in\tint(\cK)\cap\cl(\cB)$. 
	From the concavity of $f$, we have 
	\[
	f(y_t) \ge (1-t) f(y_0) + t f(y_1),\quad \liminf_{t\uparrow1} f(y_t) \ge f(y_1).
	\]
	Thus the equality $(a)$ holds. (We have shown the $\ge$ part. The $\le$ part is trivial.)
\end{proof}

\begin{lemma}\label{lem:appC2'}
	Let $\cK\not=\{0\}$ be a cone with $\cK\cap(-\cK) = \{ 0 \}$. 
	Then there exists a unit vector $e\in\cK$ such that 
	\[
	\cK\setminus\{ 0 \} \subset \set{x\in\cV | \braket{e,x}>0}.
	\]
	In particular, $\cK\cap\tint(\cK^\ast)\not=\emptyset$, 
	which also holds in the case with $\cK=\{0\}$.
\end{lemma}
\begin{proof}
	The assertion is more readily shown than Lemma~\ref{lem:appC2}. 
	Indeed, the same inequality 
	\begin{align*}
		&\max_{y\in\cK\cap\cl(\cB)} \min_{x\in\cK\cap\cS} \braket{y,x}
		\overset{(b)}{=} \max_{y\in\cK\cap\cl(\cB)} \min_{x\in\cK_0} \braket{y,x}\\
		\ge& \max_{y\in\cK_0} \min_{x\in\cK_0} \braket{y,x}
		\overset{(c)}{=} \min_{x\in\cK_0} \max_{y\in\cK_0} \braket{y,x} \ge \min_{x\in\cK_0} \|x\|^2 > 0
	\end{align*}
	holds, which implies the assertion. 
	(Note that $\cK_0\not=\emptyset$ due to $\cK\not=\{0\}$.) 
	Finally, we consider the case with $\cK=\{0\}$. 
	Since $\cK^\ast=\cV$, it follows that 
	$\cK\cap\tint(\cK^\ast) = \{0\} \not= \emptyset$.
\end{proof}

\begin{lemma}\label{lem:appC3}
	Let $\cK$ be a cone with $\cK\cap(-\cK) = \{ 0 \}$ 
	and let $x_1,x_2\in\cV$. 
	If any $\alpha\in\mathbb{R}$ satisfies $x_1 + \alpha x_2\in\cK$, 
	then $x_2=0$.
\end{lemma}
\begin{proof}
	Assume that any $\alpha\in\mathbb{R}$ satisfies $x_1 + \alpha x_2\in\cK$. 
	Then any $y\in\cK^\ast$ and $\alpha\in\mathbb{R}$ satisfy 
	\[
	0 \le \braket{y,x_1 + \alpha x_2} = \braket{y,x_1} + \alpha\braket{y,x_2},
	\]
	which implies that $\braket{y,x_2}=0$ for any $y\in\cK^\ast$. 
	Also, the relation $\tint(\cK^\ast)\not=\emptyset$ follows from Lemma~\ref{lem:appC2'}. 
	Thus Lemma~\ref{cone-span} implies that $\braket{y,x_2}=0$ for any $y\in\cV$, 
	which implies $x_2=0$.
\end{proof}

\begin{lemma}\label{state-compact}
	Let $\cK$ be a cone and let $u\in\tint(\cK^\ast)$. 
	Then $\cS(\cK,u)$ is a compact convex set.
\end{lemma}
\begin{proof}
	Since $\cS(\cK,u)$ is a closed convex set, 
	all we need to do is show that $\cS(\cK,u)$ is bounded. 
	We prove it by contradiction. 
	Suppose that some sequence $\{x_n\}_{n=1}^\infty\subset\cS(\cK,u)$ satisfied $\|x_n\|\to\infty$. 
	Then $\cK\cap\cS$ is compact, and thus 
	\[
	0 < \min_{x\in\cK\cap\cS} \braket{u,x} \le \braket{u,x_n} / \|x_n\| = 1/\|x_n\| \to 0,
	\]
	where $\cS$ is the unit sphere of $\cV$. 
	This contradiction concludes that $\cS(\cK,u)$ is bounded.
\end{proof}

\begin{lemma}
	If $\cK$ is a proper cone, 
	then $\cK^\ast$ is also a proper cone.
\end{lemma}
\begin{proof}
	From the definition, it follows that $\cK^\ast$ is a cone. 
	Also, Lemma~\ref{lem:appC2} implies $\tint(\cK^\ast)\not=\emptyset$. 
	Thus we show the remaining part $\cK^\ast\cap(-\cK^\ast) = \{ 0 \}$. 
	Let $y\in\cK^\ast\cap(-\cK^\ast)$. 
	Then $\braket{y,x}=0$ for any $x\in\cK$. 
	Moreover, Lemma~\ref{cone-span} implies that $\braket{y,x}=0$ for any $x\in\cV$, 
	which implies $y=0$.
\end{proof}

\begin{lemma}\label{cone-interior}
	Let $\cK_1$ and $\cK_2$ be proper cones of $\cV_1$ and $\cV_2$ respectively, 
	$\tilde{\cK}$ be a convex cone with 
	$\tilde{\cK}_{\min} \subset \tilde{\cK} \subset \tilde{\cK}_{\max}$, 
	and let $x_i\in\cK_i$ for $i=1,2$. 
	Then $x_1\otimes x_2\in\tint(\tilde{\cK})$ if and only if $x_i\in\tint(\cK_i)$ for each $i=1,2$.
\end{lemma}
\begin{proof}
	Let us show the only if part. 
	Suppose $x_1\otimes x_2\in\tint(\tilde{\cK})$. 
	Take arbitrary $y_1\in\cK_1^\ast\setminus\{ 0 \}$ and 
	$y_2\in\cK_2^\ast\setminus\{ 0 \}$. 
	Then $x_1\otimes x_2\in\tint(\tilde{\cK})\subset\tint(\cl(\tilde{\cK}))$ 
	and $y_1\otimes y_2 \in \tilde{\cK}_{\max}^\ast\setminus\{ 0 \} \subset \tilde{\cK}^\ast\setminus\{ 0 \} = \cl(\tilde{\cK})^\ast\setminus\{ 0 \}$. 
	Thus Lemma~\ref{lem:appC1} implies 
	\[
	0 < \braket{x_1\otimes x_2, y_1\otimes y_2} = \braket{x_1,y_1}\braket{x_2,y_2},
	\]
	whence $\braket{x_1,y_1}>0$ and $\braket{x_2,y_2}>0$. 
	Since $y_1\in\cK_1^\ast\setminus\{ 0 \}$ and 
	$y_2\in\cK_2^\ast\setminus\{ 0 \}$ are arbitrary, 
	Lemma~\ref{lem:appC1} implies $x_1\in\tint(\cK_1)$ and $x_2\in\tint(\cK_2)$.
	\par
	Conversely, letting $x_i\in\tint(\cK_i)$ for $i=1,2$, 
	we can take a basis $\{x_{i,j}\}_{j=1}^{d_i}\subset\cK_i$ of $\cV_i$ 
	with $\sum_{j=1}^{d_i} x_{i,j} = x_i$ for each $i=1,2$. 
	Since the tuple $\{x_{1,j_1}\otimes x_{2,j_2}\}_{j_1,j_2}\subset\tilde{\cK}_{\min}$ is a basis of $\cV_1\otimes\cV_2$, 
	the set 
	\[
	\tilde{\cO} \coloneqq \Set{\sum_{j_1,j_2} \alpha_{j_1,j_2} x_{1,j_1}\otimes x_{2,j_2} | \alpha_{j_1,j_2}>0\ (\forall j_1,j_2)}
	\subset \tilde{\cK}_{\min}\subset\tilde{\cK}
	\]
	is an open set of $\cV_1\otimes\cV_2$. 
	Thus $x_1\otimes x_2\in\tilde{\cO}\subset\tint(\tilde{\cK})$.
\end{proof}

\begin{lemma}
	If $\cK_1$ and $\cK_2$ are proper cones of $\cV_1$ and $\cV_2$ respectively, 
	then $\tilde{\cK}_{\min}$ is also a proper cone of $\cV_1\otimes\cV_2$.
\end{lemma}
\begin{proof}
	From the definition, it follows that $\tilde{\cK}_{\min}$ is a convex cone. 
	Also, Lemma~\ref{cone-interior} implies $\tint(\tilde{\cK}_{\min})\not=\emptyset$ 
	because there exist $u_1\in\tint(\cK_1)$ and $u_2\in\tint(\cK_2)$.
	\par
	In order to prove that $\tilde{\cK}_{\min}$ is closed, 
	we show that $\cS(\tilde{\cK}_{\min}, u_1\otimes u_2)$ is compact. 
	Lemma~\ref{state-compact} yields that $\cS(\cK_1,u_1)$ and $\cS(\cK_2,u_2)$ are compact 
	and so is $\cS(\cK_1,u_1)\times\cS(\cK_2,u_2)$, 
	where $\cX_1\times\cX_2$ denotes the direct product of two sets $\cX_1$ and $\cX_2$. 
	Defining the map $f: \cV_1\times\cV_2\to\cV_1\otimes\cV_2$ 
	to be $(x_1,x_2)\mapsto x_1\otimes x_2$, 
	we find that the map $f$ is continuous, 
	whence $f(\cS(\cK_1,u_1)\times\cS(\cK_2,u_2))$ is compact. 
	Thus $\conv\bigl( f(\cS(\cK_1,u_1)\times\cS(\cK_2,u_2)) \bigr)$ is also compact. 
	Since the relation 
	\[
	\conv\bigl( f(\cS(\cK_1,u_1)\times\cS(\cK_2,u_2)) \bigr) = \cS(\tilde{\cK}_{\min}, u_1\otimes u_2)
	\]
	follows from the definitions, 
	the set $\cS(\tilde{\cK}_{\min}, u_1\otimes u_2)$ is compact.
	\par
	Let $\{\tilde{x}_n\}_{n=1}^\infty\subset\tilde{\cK}_{\min}$ be a convergent sequence 
	whose limit is $\tilde{x}\in\cV_1\otimes\cV_2$. 
	We show $\tilde{x}\in\tilde{\cK}_{\min}$. 
	Without loss of generality, 
	we may assume $\tilde{x}\not=0$ due to $0\in\tilde{\cK}_{\min}$. 
	Since $\tilde{x}\not=0$ and $u_1\otimes u_2\in\tint(\tilde{\cK}_{\min})$, 
	Lemma~\ref{lem:appC1} implies that 
	any sufficiently large $n$ satisfies $\braket{u_1\otimes u_2, \tilde{x}_n} > 0$. 
	Since $\cS(\tilde{\cK}_{\min}, u_1\otimes u_2)$ is compact, 
	there exists a subsequence $\{x_{n(k)}\}_{k=1}^\infty$ of $\{x_n\}_{n=1}^\infty$ such that 
	the sequence $\{\braket{u_1\otimes u_2, \tilde{x}_{n(k)}}^{-1} \tilde{x}_{n(k)}\}_{k=1}^\infty$ converges to 
	a state $\tilde{x}_0\in\cS(\tilde{\cK}_{\min}, u_1\otimes u_2)$. 
	Since $\braket{u_1\otimes u_2, \tilde{x}_{n(k)}}\to\braket{u_1\otimes u_2, \tilde{x}}\ge0$, 
	we have $\tilde{x}_{n(k)}\to \tilde{x} = \braket{u_1\otimes u_2, \tilde{x}}\tilde{x}_0\in\tilde{\cK}_{\min}$. 
	Therefore, $\tilde{\cK}_{\min}$ is closed.
	\par
	Finally, we show $\tilde{\cK}_{\min}\cap(-\tilde{\cK}_{\min}) = \{ 0 \}$. 
	Let $\tilde{x}\in\tilde{\cK}_{\min}\cap(-\tilde{\cK}_{\min})$ and $u_i\in\tint(\cK_i^\ast)$ for $i=1,2$. 
	If $\tilde{x}\not=0$, 
	Lemmas~\ref{lem:appC1} and \ref{cone-interior} imply 
	$0 < \braket{u_1\otimes u_2, \tilde{x}} < 0$, 
	which is a contradiction. 
	Therefore, $\tilde{x}=0$. 
	Since $0\in\tilde{\cK}_{\min}\cap(-\tilde{\cK}_{\min})$, 
	we obtain $\tilde{\cK}_{\min}\cap(-\tilde{\cK}_{\min}) = \{ 0 \}$.
\end{proof}

\end{document}